\RequirePackage{fix-cm}
\RequirePackage{fixltx2e}
\documentclass[reqno]{amsproc}
\usepackage{amsthm,amssymb}
\usepackage{mathtools}
\usepackage{amsmath}

\usepackage[citation-order,nobysame]{amsrefs}
\usepackage{xyzbib}

\usepackage{pstricks,pst-node}
\usepackage{hyperref}
\hypersetup{linktocpage=true}

\mathtoolsset{showonlyrefs}
\numberwithin{equation}{section}
\let\cite=\cites
\newcommand{\rmd}{\mathrm{d}}
\newcommand{\rmi}{\mathrm{i}}
\newcommand{\rme}{\mathrm{e}}

\newcommand{\vc}{u}

\DeclareMathOperator{\Tr}{\mathrm{Tr}}
\DeclareMathOperator{\Det}{\mathrm{Det}}
\DeclareMathOperator{\Ker}{\mathrm{Ker}}

\newcommand{\Ftwoone}[4]{%
\,{}_{2}F_{1}\bigg(\genfrac{}{}{0pt}{}{#1,\, #2}{#3} \bigg\vert #4\bigg)}
\newcommand{\Fthreetwo}[6]{%
\,{}_{3}F_{2}\bigg(\genfrac{}{}{0pt}{}{#1,\, #2,\, #3}{#4,\, #5} \bigg\vert #6\bigg)}

\DeclareMathOperator{\tr}{\mathrm{tr}}
\DeclareMathOperator*{\res}{\mathrm{res}}

\newtheorem{theorem}{Theorem}[section]
\newtheorem{lemma}{Lemma}[section]
\newtheorem{proposition}{Proposition}[section]
\newtheorem{corollary}{Corollary}[section]
\newtheorem{remark}{Remark}[section]
\begin{document}
\vspace*{.5in}
\title
{Emptiness formation probability of the six-vertex model and
the sixth Painlev\'e equation}

\thanks{This work is partially supported by the Russian Foundation for Basic Research,
Grant No. 13-01-00336}

\author{A. V. Kitaev}
\author{A. G. Pronko}

\address{V. A. Steklov Mathematical Institute,
Fontanka 27, 191023 St.~Petersburg, Russia}

\email{kitaev@pdmi.ras.ru}
\email{agp@pdmi.ras.ru}

\begin{abstract}

We show that the emptiness formation probability
of the six-vertex model  with domain wall boundary conditions
at its free-fermion point is
a $\tau$-function of the sixth Painlev\'e equation. Using this fact
we derive asymptotics of the emptiness formation probability
in the thermodynamic limit.

\end{abstract}

\subjclass[2000]{34E05, 82B23, 34M55, 33E17, 82B26}
\keywords{Fredholm determinants, Hankel determinants,
Fuchs pair, lattice models, correlation functions, 
asymptotic expansions, saddle-point method}
\maketitle
\tableofcontents
\section{Introduction}

One of the most intriguing facts about correlation functions
of solvable models of statistical mechanics is that in many cases
they can be described in terms of of the Painlev\'e equations or their
generalizations. Famous examples are provided by correlation functions
of the two-dimensional Ising model, related to the third
Painleve equation \cite{BMcCW-73,WMcCTB-73}, and correlation functions
of an impenetrable Bose gas, related to the fifth Painlev\'e equation \cite{SMJ-79,JM-81}.
At the same time this relationship carries mainly academic interest, whilst
asymptotic analysis of the correlation functions is usually performed via the associate integrable
structures rather than with the help of these ordinary differential equations (ODEs).

In this paper, we consider the six-vertex model with the domain wall boundary conditions
and discuss a particular correlation function, called
the emptiness formation probability (EFP). We show that, for the model
with the Boltzmann weights satisfying the free-fermion condition, this
correlation function appears to be a $\tau$-function of the
sixth Painlev\'e equation (P6). Using this connection we
derive asymptotic expansions of the EFP in the thermodynamic limit.

It is important to stress that here we consider asymptotics of P6 for the case where the coefficients
of the equation are large while
its argument is a finite parameter. Such type of asymptotics
is not thoroughly studied for P6 yet. Therefore to cope with this problem,
we initially intended to use one of the known asymptotic methods
for the Painlev\'e equations, namely, the isomonodromy deformation technique
by Jimbo-Its-Novokshenov \cite{J-82,IN-86} or Deift-Zhou \cite{DZh-93} asymptotic
analysis of the corresponding Riemann-Hilbert problem.
However, we were quite surprised that we have been able to construct
asymptotics (both the leading terms and corrections) without
any sophisticated techniques---just
with the help of the $\sigma$-form of P6. Indeed,
finding asymptotics, especially the correction terms, of solutions
of ODEs by substituting corresponding asymptotic expansions
into the ODE is the standard asymptotic technique
in the case where the asymptotics with respect to the
argument of the corresponding ODE is constructed. On this way, one
usually finds a recurrence relation for determination
of the coefficients. Our case is different: we arrive at the recurrence relation
for the first derivatives of the correction terms
so that to find the explicit formulas we need the initial data. It turn out that
for correction terms of arbitrary order these initial data can be obtained from
representations of the EFP at the critical points of P6. It seems that the same technique
may appear useful for other analogous problems.

\subsection{The model}

We recall that the six-vertex model (also known as the ice-type model)
is a statistical mechanics model defined on a square lattice.
The local states of the model are arrows placed on the edges of the lattice.
The admissible configurations of the model are those in which there are equal numbers
of incoming and outgoing arrows in each lattice vertex. The condition
of local conservation of the number of incoming and outgoing arrows
is known as the ice-rule, and it selects six possible arrow configurations
around the vertex, see Fig.~\ref{fig-sixv}.
For periodic boundary conditions the model was solved by Lieb \cite{L-67a,L-67b,L-67c}
and Sutherland \cite{S-67}
(for a review, see, e.g., \cite{B-82}).

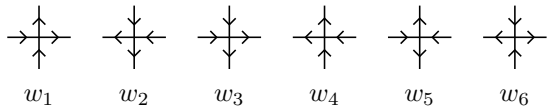
\begin{figure}
\centering

\psset{unit=12pt,linewidth=.05,dotsep=2pt}
\newcommand{\arr}{\lput{:U}{\begin{pspicture}(0,0)
\psline(-.1,.2)(.1,0)(-.1,-.2) \end{pspicture}}}

\begin{pspicture}(0,-1)(17,2)
\rput(0,0){
\pcline(0,1)(1,1)\arr \pcline(1,1)(2,1)\arr
\pcline(1,0)(1,1)\arr \pcline(1,1)(1,2)\arr}
\rput(3,0){
\pcline(1,1)(0,1)\arr \pcline(2,1)(1,1)\arr
\pcline(1,2)(1,1)\arr \pcline(1,1)(1,0)\arr}
\rput(6,0){
\pcline(0,1)(1,1)\arr \pcline(1,1)(2,1)\arr
\pcline(1,2)(1,1)\arr \pcline(1,1)(1,0)\arr}
\rput(9,0){
\pcline(2,1)(1,1)\arr \pcline(1,1)(0,1)\arr
\pcline(1,0)(1,1)\arr \pcline(1,1)(1,2)\arr}
\rput(12,0){
\pcline(0,1)(1,1)\arr \pcline(2,1)(1,1)\arr
\pcline(1,1)(1,0)\arr \pcline(1,1)(1,2)\arr}
\rput(15,0){
\pcline(1,1)(2,1)\arr \pcline(1,1)(0,1)\arr
\pcline(1,2)(1,1)\arr \pcline(1,0)(1,1)\arr}
\rput[B](1,-1){$w_1$}
\rput[B](4,-1){$w_2$}
\rput[B](7,-1){$w_3$}
\rput[B](10,-1){$w_4$}
\rput[B](13,-1){$w_5$}
\rput[B](16,-1){$w_6$}
\end{pspicture}

\caption{The six types of arrow configurations around a vertex
allowed in the six-vertex model and their Boltzmann weights.}
\label{fig-sixv}
\end{figure}

Domain wall boundary conditions (DWBC) are defined as follows.
Consider a finite square lattice formed by intersection of
$N$ horizontal and $N$ vertical lines (the so-called $N\times N$ lattice).
For such a lattice one may require that the local states on external edges are
fixed in a special way consistent with the ice-rule, namely, using the
description  of the local states in terms of arrows on edges,
that all arrows on the external
vertical edges (i.e., at the top and bottom boundaries) are incoming, while
those on the external horizontal edges (i.e., on the left and right boundaries)
are outgoing, see Fig.~\ref{fig-frsq}.
The six-vertex model on the $N\times N$ lattice with such fixed boundary
conditions is called the six-vertex model with DWBC;
it was considered for the first time by Korepin \cite{K-82}.

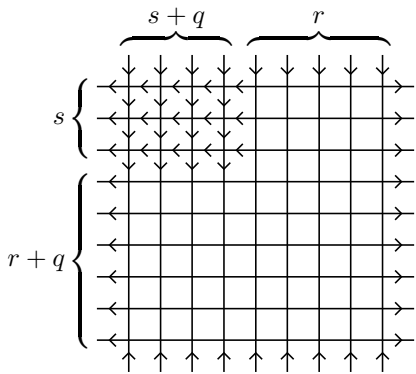
\begin{figure}
\centering

\psset{unit=12pt,linewidth=.05}
\newcommand{\arr}{\lput{:U}{\begin{pspicture}(0,0)
\psline(-.1,.2)(.1,0)(-.1,-.2) \end{pspicture}}}

\begin{pspicture}(-2.5,0)(10,11.5)
%
\multirput(1,0)(1,0){9}{\pcline(0,0)(0,1)\arr \pcline(0,10)(0,9)\arr}
\multirput(0,1)(0,1){9}{\pcline(1,0)(0,0)\arr \pcline(9,0)(10,0)\arr}
\multirput(1,8)(1,0){4}{\pcline(0,1)(0,0)\arr}
\multirput(1,7)(1,0){4}{\pcline(0,1)(0,0)\arr}
\multirput(1,6)(1,0){4}{\pcline(0,1)(0,0)\arr}

\multirput(0,9)(0,-1){3}{\pcline(2,0)(1,0)\arr}
\multirput(0,9)(0,-1){3}{\pcline(3,0)(2,0)\arr}
\multirput(0,9)(0,-1){3}{\pcline(4,0)(3,0)\arr}
\multirput(0,9)(0,-1){3}{\pcline(5,0)(4,0)\arr}

\multirput(5,7)(0,1){3}{\psline(0,0)(4,0)}
\multirput(1,1)(0,1){6}{\psline(0,0)(8,0)}

\multirput(1,1)(1,0){4}{\psline(0,0)(0,5)}
\multirput(5,1)(1,0){5}{\psline(0,0)(0,8)}


\rput{-90}(-.5,8)
  {$\underbrace{\begin{pspicture}(.75,0)(3.25,0)\end{pspicture}}$}
\rput[r](-1,8){$s$}

\rput{-90}(-.5,3.5)
  {$\underbrace{\begin{pspicture}(.25,0)(5.75,0)\end{pspicture}}$}
\rput[r](-1,3.5){$r+q$}

\rput{-180}(7,10.5)
  {$\underbrace{\begin{pspicture}(0.25,0)(4.75,0)\end{pspicture}}$}
\rput[B](7,11){$r$}

\rput{-180}(2.5,10.5)
  {$\underbrace{\begin{pspicture}(0.25,0)(3.75,0)\end{pspicture}}$}
\rput[B](2.5,11){$s+q$}

\end{pspicture}

\caption{An $N\times N$ lattice ($N=r+s+q$) with domain wall boundary conditions, with
the additional condition that the vertices of an $s\times (s+q)$
lattice at the top left corner are all of type $2$---the probability
of this configuration is $F_{r,s,q}$ (shown the case $r=5$, $s=3$, and $q=1$).}
\label{fig-frsq}
\end{figure}

To define the partition function of this model, let us denote by $\Omega_N$
the set of all arrow configurations of the $N\times N$ lattice with
DWBC and with the ice-rule obeyed in each lattice vertex.
Assigning to the $i$th arrow configuration around a vertex the Boltzmann weight
$w_i$, see Fig.~\ref{fig-sixv}, the partition function $Z_N$ is defined as
\begin{equation}
Z_N=\sum_{\mathcal{C}\in\Omega_N}\prod_{i=1,\dots,6} w_i^{n_i(\mathcal{C})}.
\end{equation}
Here $n_i(\mathcal{C})$, $i=1,\dots,6$, is the number of vertices
of type $i$ in the configuration $\mathcal{C}$, $\sum_i n_i(\mathcal{C})=N^2$.
The partition function was computed, for a more general inhomogeneous
(with position-dependent weights)
model, in terms of an $N\times N$ determinant by Izergin \cite{I-87};
for the homogeneous model
the determinant has the Hankel structure \cite{ICK-92}.
The free energy per site was derived for various regimes
and by various methods by Korepin and Zinn-Justin \cite{KZj-00,Zj-00}.
A detailed analysis of the asymptotic expansion
of the partition function in the thermodynamic limit, using
the connection of Hankel determinants and matrix models
with the Riemann-Hilbert problem
\cite{FIK-93}, was given in paper \cite{BKZ-02} and
a series of papers by Bleher and
collaborators \cite{BF-06,BL-09a,BL-09b,BL-10,BB-12}.

\subsection{Emptiness formation probability}
To define the EFP, let us
set $N=r+s+q$, where $r$, $s$, and $q$ are nonnegative integers,
and consider configurations
of six-vertex model on the $N\times N$ lattice with DWBC
in which the vertices belonging to the
$s\times (s+q)$ rectangle at the top left corner of the lattice
are all of type 2, see Figs.~\ref{fig-sixv} and \ref{fig-frsq}.
We denote the set of such configurations as $\Omega_{r,s,q}$;
obviously, $\Omega_{r,s,q}\subset \Omega_N$. We denote the EFP
by $F_{r,s,q}$ and define it as a probability of observing a configuration
belonging to $\Omega_{r,s,q}$ for a given
$r,s,q$,
\begin{equation}
F_{r,s,q}=Z_N^{-1}
\sum_{\mathcal{C}\in\Omega_{r,s,q}}\prod_{i=1,\dots,6} w_i^{n_i(\mathcal{C})}.
\end{equation}
In what follows we always assume that
$s\leq r$, since otherwise $F_{r,s,q}$ is just zero
(for $s>r$ the set $\Omega_{r,s,q}$ is empty).
A closed expression for the EFP
in the terms of a multiple ($s$-fold) integral was obtained by Colomo and the
second author in \cite{CP-07b}. As it was demonstrated in \cite{CP-09}, this
result make it possible to address phase separation phenomena in the model,
and in one of its special cases provides the limit shape of large altering-sign
matrices \cite{CP-08}.

In this paper we study the EFP of the six-vertex model with DWBC in
the case of the weights obeying the so-called
free-fermion condition, $w_1w_2+w_3w_4=w_5w_6$. Namely, we choose
the weights in the form
\begin{equation}\label{FF}
w_1=w_2=\sqrt{1-\alpha},\qquad
w_3=w_4=\sqrt{\alpha},\qquad
w_5=w_6=1,
\end{equation}
where $\alpha\in(0,1)$ is a parameter. We note that this case is
related to enumeration of domino tilings, where
the parameter $\alpha$ plays the role of ``bias'' \cite{EKLP-92,CEP-96,CP-13}.
In the main text we prove the following result.
\begin{theorem}\label{EFPasP6}
Define
\begin{equation}\label{eq:sigma-frsq}
\sigma=\alpha(\alpha-1)
\frac{\rmd}{\rmd\alpha}\log F_{r,s,q}
-\frac{(r+q+s)^2}{4}\alpha+\frac{(r+q+s)q+2rs}{4},
\end{equation}
then
\begin{multline}\label{sigmaform}
\alpha^2(\alpha-1)^2\sigma'(\sigma'')^2
+\left\{(1-2\alpha)(\sigma')^2+2\sigma\sigma'+\nu_1\nu_2\nu_3\nu_4\right\}^2
\\
=\left(\sigma'+\nu_1^2\right)
\left(\sigma'+\nu_2^2\right)
\left(\sigma'+\nu_3^2\right)
\left(\sigma'+\nu_4^2\right),\qquad
\sigma'\equiv\frac{\rmd}{\rmd\alpha}\sigma,
\end{multline}
where
\begin{equation}
\nu_1=\nu_3=-\frac{r+q+s}{2},\quad
\nu_2=-\frac{r-q-s}{2},\quad
\nu_4=\frac{r+q-s}{2}.
\end{equation}
\end{theorem}

\begin{remark}
Note that $F_{r,s,q}$ is a polynomial in $\alpha$, see for details
representation \eqref{newEFP} obtained in Sect.~\ref{Sec:EFP}.
Thus, \eqref{eq:sigma-frsq} implies that $\sigma=\sigma(\alpha)$
is a rational solution of \eqref{sigmaform}.
The corresponding solution of the canonical form of $P6$ is also rational.
This solution defines isomonodromy deformations
of the associated linear Fuchsian ODE with the monodromy matrices $\pm I$
(see Sect.~\ref{Sec:FDR} for details).
\end{remark}

\begin{remark}
The EFP can be expressed as
\begin{equation}\label{eq:frsq-tau}
F_{r,s,q}=C_{r,s,q}\alpha^{rs}(1-\alpha)^{r(r+q)}\tau,
\end{equation}
where quantity $C_{r,s,q}$ is independent of $\alpha$ and $\tau=\tau(\alpha)$ is the
Jimbo-Miwa $\tau$-function of $P6$ \cite{JM-81},
see also Sect.~\ref{Sec:tau-f}.
The pre-factor of the $\tau$-function
in \eqref{eq:frsq-tau} is a matter of the definition
(normalization) of the $\tau$-function,
therefore $F_{r,s,q}$ is, in fact, the $\tau$-function of $P6$.
\end{remark}

We recall some basic fact about the theory of P6 in Sect.~\ref{Sec:tau-f}.
Our starting point in proving the theorem is the formula for the EFP
expressing it as a Fredholm determinant of some linear integral operator
of the so-called integrable type. This representation, among some others,
was derived from the general multiple integral representation
by specifying it to the case of the weights \eqref{FF} in \cite{P-13}.
We review some of these results in Sect.~\ref{Sec:EFP}. The
theorem is proven in Sects.~\ref{Sec:2} and \ref{Sec:3}.

As a comment to this result, we mention that appearance of a Painlev\'e
equation in the context of the six-vertex model with DWBC was already noticed
in \cite{ST-10} where it was shown that the partition function
in the ferroelectric regime serves as a $\tau$-function of the fifth
Painlev\'e equation.
In turn, P6 appeared in the study of the
EFP at zero-temperature for the Heisenberg XX0 spin chain \cite{DIZ-97},
which can be seen as the six-vertex model at its free-fermion
point with periodic boundary conditions.

In fact, the relation of $F_{r,s,q}$ with $P6$ reported
in Theorem~\ref{EFPasP6} can be extracted from the works by Johansson, Baik and Rains,
and Forrester and Witte
\cite{J-00,BR-01,FW-04} (see a discussion at the end of Sect.~\ref{Sec:EFP}).
Here we present a different proof
based on Fredholm determinants of integrable
linear integral operators, developed
in \cite{SMJ-79,IIKS-90}. This method is also
important for asymptotic study of $F_{r,s,q}$ as it can be used
to provide initial conditions in our proof of Theorem~\ref{TDL2} below.
The connection of the EFP with
P6 allows us to study behavior of the EFP in the thermodynamic limit,
in which the integers $r$, $s$, and $q$ are large, with their ratios fixed;
the variable $\alpha$, defining the weights,
is considered as a parameter. An important special case of such a limit is where the
region of the vertices of type 2 at the top left corner
has a macroscopically square shape, that corresponds to the ratios $q/s$ and $q/r$
vanishing in the limit.  In this paper we compute the asymptotic
expansion for $F_{r,s,q}$ in the case of $q=0$.

\subsection{Main result: asymptotic expansions}

Define
\begin{equation}
v\equiv\frac{s}{r},\quad
\vc\equiv\frac{1-\sqrt\alpha}{1+\sqrt\alpha},\qquad v,u\in (0,1).
\end{equation}
As discussed in Sect.~\ref{Sec:tworegimes}, the EFP has different asymptotic expansions
for $u< v < 1$ (referred to as the EFP in the disordered regime)
and for $0< v < u$ (referred to as the EFP in the ordered regime).
Our results about the EFP in the thermodynamic limit are formulated as follows.
\begin{theorem}\label{TDL}
For $s/r=v\in(\vc,1)$, as $s,r\to\infty$,
\begin{multline}\label{lowertail}
\log F_{r,s,0}
=-\varphi s^2- \frac{1}{12}\log s+
\frac{1}{8}\log\frac{(1-\vc^2)v^2}{v^2-\vc^2}-\frac{1}{12}\log\frac{1-v^2}{2}
+\zeta'(-1)
\\
+\sum_{k=1}^n a_{2k}s^{-2k} +O(s^{-2n-2}),
\end{multline}
where
\begin{equation}\label{varphi}
\varphi=\log\frac{v}{\vc}-\frac{(1-v)^2}{2v^2}\log\frac{1-v}{1-\vc}
-\frac{(1+v)^2}{2v^2}\log\frac{1+v}{1+\vc},
\end{equation}
$\zeta'(-1)=-0.1654211437\ldots$ is the derivative of the Riemann
zeta-function, and $a_{2k}=a_{2k}(u,v)$
are rational functions of $u^2$ and $v^2$ with poles at $v^2=1$ and 
$v^2=u^2$. In particular,
\begin{align}
a_2&
=\frac{u^2(1-v^2)(2v^4+5u^2v^2-u^4)}{64(v^2-u^2)^3}
-\frac{(1+v^2)(v^2-(1-v^2)^2)}{120(1-v^2)^2}-\frac{1}{64},
\notag\\
a_4&
=-\frac{u^2(1-v^2)}{256(v^2-u^2)^6}\big[10v^{10}u^2-2v^{10}-90v^8u^2
\notag\\ &\quad
+140v^8u^4+105v^6u^6-160v^6u^4-4v^4u^8
+5v^4u^6-6v^2u^8+v^2u^{10}+u^{10}\big]
\notag\\ &\quad
-\frac{v^6(v^6-4v^4+5v^2-10)}{504(1-v^2)^4}+\frac{31}{16128}.
\end{align}
The estimate $O(s^{-2n-2})$ 
of the error term in \eqref{lowertail} is uniform with respect to the parameter 
$v$ on any compact subset of the interval $(u, 1)$.
\end{theorem}

\begin{theorem}\label{TDL2}
For $s/r=v\in(0,u)$, as $s,r\to\infty$,
\begin{multline}\label{uppertail}
\log \left(1-F_{r,s,0}\right)
=-\chi s- \log s
+\log\frac{v^2(1-\vc^2)^2}{32\pi\sqrt{\vc}[(\vc-v) (1-\vc v)]^{3/2}}
\\
+\sum_{k=1}^n b_{k}s^{-k} +O(s^{-n-1}),
\end{multline}
where
\begin{equation}\label{chiuv}
\chi=\frac{4}{v}
\log\left(\frac{\sqrt{1-v\vc}+\sqrt{\vc(\vc-v)}}{\sqrt{1-\vc^2}}\right)
-4\log \left(\frac{\sqrt{\vc(1-v\vc)}+\sqrt{\vc-v}}{\sqrt{\left(1-\vc^2\right)v}}
\right)
\end{equation}
and $b_k=b_k(u,v)$. In particular,
\begin{align}
b_1&=-\frac{9v^2(1+\vc^4)-36v\vc(1+v^2)(1+\vc^2)-(2-142v^2+2v^4)\vc^2}
{48\sqrt{\vc}[(\vc-v) (1-\vc v)]^{3/2}},
\notag\\
b_2&=\frac{v^2}{64\vc(\vc-v)^3 (1-v\vc)^3}
\big[3v^2(\vc^8+1)-8v \vc(v^2+1)(\vc^6+1)
\notag\\ &\quad
+4\vc^2(10v^4+v^2+10)(\vc^4+1)-120v\vc^3(v^2+1)(\vc^2+1)
\\ &\quad
-(16v^4-370v^2+16)\vc^4\big].
\end{align}
The estimate $O(s^{-n-1})$ of the error term in \eqref{uppertail} is uniform
with respect to the parameter $v$ on any compact subset of the 
interval $(0, u)$.
\end{theorem}
\begin{remark}
Since $s$ and $r$ are integers $v$ is a rational number,
$v=s_0/r_0$, where $0<s_0<r_0$ are co-prime natural numbers.
Asymptotic expansions presented in Theorems~\ref{TDL} and \ref{TDL2}
are understood in the natural way: put $s\equiv s_0\,n$ and $r\equiv r_0\,n$,
where $n$ is a positive integer and $n\to\infty$.
As a further comment, we note that
one can modify the problem and ask about asymptotics for irrational values of $v$.
Our results also deliver the answer to this question.
Consider rational approximation, $v=s_n/r_n+\delta_n$,
where $s_n$ and $r_n$ are co-prime natural numbers, and $\delta_n\to0$
is the corresponding error.
With the help of the Farey series \cite{HW-08} we can find such approximations
that $\delta_n=\kappa_n/r_n^2$,
where $|\kappa_n|<1$. For the irrational $v$, Theorems~\ref{TDL} and \ref{TDL2}
remain valid but with
$v$ substituted by $v-\delta_n$. In particular, for
the optimal rational approximations (obtained with the help of the
continued fractions or the Farey series)
it is convenient first to rewrite the asymptotics in terms
of the large parameter $r$ instead of $s$ by using the relation $s=vr$,
then substitute $v$ with $v-\kappa_n/r_n^2$, and, finally,
re-expand asymptotics with respect to $r=r_n$.
\end{remark}

\begin{remark}
Existence of the asymptotic expansions \eqref{uppertail} and \eqref{lowertail}
are proved in Sect.~\ref{Sec:Disorder} and Sect.~\ref{Sec:Order}, respectively,
and the coefficients $a_{2k}$'s and $b_k$'s can be obtained by recurrent procedures via
substitution of these expansions into the $\sigma$-form of P6, see 
Theorem \ref{EFPasP6}.
The coefficients are derived successively with the linear differential
equations of the first order whose initial data are given in
these sections.
\end{remark}

The asymptotic formulas \eqref{lowertail} and \eqref{uppertail} confirms
and extends the result of \cite{CP-13}, where
only the leading $O(s^2)$-term of the asymptotic expansion of $\log F_{r,s,0}$
was obtained. Formula \eqref{chiuv} for $\chi$ reproduces the result obtained by
Johansson~\cite{J-00}, see, for details, Sect.~\ref{Sec:Order}.

As a comment to the asymptotic expansion of the EFP, we mention that
the value $v=\vc$ is the critical value at which the third-order
phase transition of the leading term takes place. This phase transition has
a combinatorial interpretation in the context the
of random domino tilings on Aztec diamonds with a cut-off
corner \cite{CP-13}.
Recently, a possible third-order phase transition
in the thermal correlations of XX0 spin chain, although for a
peculiar choice of the parameters, was pointed out in
\cite{PgT-14}. Similar phenomena can be seen
when the EFPs of quantum spin chains and
vertex models are treated by conformal field theory methods \cite{S-14}
and when the vector Chern-Simons theory is considered on $S^2\times S^1$
\cite{JMSTWY-13}.
A detailed discussion of the third order phase transition phenomena
from the random matrix models point-of-view can be found in \cite{MS-14}.
We discuss aspects of the third-order phase transition related to the EFP
in terms of the associated Painlev\'e equations in Conclusion.

\section{EFP and integrable structures}\label{Sec:2}

The aim of this section is  to present an overview of the results
which lead to formulation of the EFP as a special solution of an
integrable system (in our case P6). We start with recalling presentation of
$F_{r,s,q}$ as a Fredholm determinant of a
linear integral operator, next we discuss the properties of the corresponding
resolvent kernel, and finally we show how to convert
the associated system of integral equations to
Fuchs (a Lax-type) pair for P6.

\subsection{EFP as a Fredholm determinant}\label{Sec:EFP}

We mention here
three representations for the EFP, valid
for the six-vertex model with the weights \eqref{FF}:
in terms of a multiple integral, in terms of Hankel determinant, and
in terms a Fredholm determinant of a linear integral operator acting on
a closed contour in the complex plane. Others can be found in \cite{P-13}
where interrelations between these representations are proved.
In our study of the EFP we mostly rely on the last form, in terms of the
Fredholm determinant.

In \cite{CP-07b}, using the method of Yang-Baxter algebra
(the algebra of the quantum monodromy matrix) \cite{KBI-93} and the methods
of orthogonal polynomials theory, the EFP was evaluated in the form of a multiple
integral. In the case of the weights \eqref{FF} the result can be formulated
as follows.

\begin{proposition}
The EFP of the six-vertex model with DWBC with the weights \eqref{FF} is represented
in the form
\begin{equation}\label{oldEFP}
F_{r,s,q}=
\frac{(-1)^{s(s+1)/2}}{s!}
\oint_{C_0}\cdots\oint_{C_0}\prod_{1\leq j<k\leq s}(z_j-z_k)^2
\prod_{j=1}^{s}
\frac{(\alpha z_j+1-\alpha)^{r+q}}{z_j^r (z_j-1)^s} \frac{\rmd^s z}{(2\pi\rmi)^s},
\end{equation}
where $C_0$ is a simple, closed, counter-clockwise oriented contour around
the point $z=0$ of the complex plane, and lying  in its small vicinity.
\end{proposition}

The representation \eqref{oldEFP} was used in the study of the so-called arctic curve
of the six-vertex model with DWBC, i.e., the curve which describes
the spacial separation of order and disorder in the thermodynamic limit
\cite{CP-07a,CP-09}. However, the
integral \eqref{oldEFP} can hardly be studied by the random
matrix model methods, since the problem of finding of an equilibrium measure
cannot be solved explicitly (see discussion in \cite{CP-07a}).
This fact stimulated to look for another, simple representations.
In particular, evaluating the integrals it can be shown
that the EFP is given in terms of a Hankel determinant.

\begin{proposition}
The EFP admits the representation:
\begin{equation}\label{newEFP}
F_{r,s,q}=\frac{(q!)^s}{\prod_{k=0}^{s-1} (q+k)! k!}
\frac{(1-\alpha)^{s (s+q)}}{\alpha^{s(s-1)/2}}
\det_{1\leq j,k\leq s}\left[ \sum_{m=0}^{r-1}m^{j+k-2} \binom{m+q}{m} \alpha^m\right].
\end{equation}
\end{proposition}

Using a peculiar structure of the matrix in \eqref{newEFP},
the finite-size determinant here can converted into Fredholm determinants of
integral operators of various types. Indeed, the matrix in \eqref{newEFP}
can be represented as a difference of two matrices,
one involving the summation over $m\in\mathbb{Z}_{\geq 0}$
and second over $m\in\mathbb{Z}_{\geq r}$; the first matrix
is related with the ensemble of Meixner polynomials and hence it can be
explicitly inverted, while the second one gives rise to a Fredholm determinant
structure of the finite-size determinant. This finite-size determinant can be further
transformed by keeping the Fredholm structure of the determinant.
In particular, it can be written
as a Fredholm determinant of a linear integral operator acting on a closed contour
in the complex plane. Specifically, here we deal with the integral operator $\hat K$
acting on trial functions by the formula
\begin{equation}
\left(\hat K f\right)(\lambda)=\oint_{C_0} K(\lambda,\mu) f(\mu)\, \rmd \mu.
\end{equation}
The function $K(x,y)$ is the kernel of the operator,
and $C_0$ is the contour given as a circle of small radius around the point
$\lambda=0$, with counter-clockwise orientation, i.e., the same contour
appeared above in the multiple integral representation for the EFP.
Besides this contour, the functions defining the kernel
involve integrals over the contour $C_\infty$, which is
a circle of large radius, counter-clockwise oriented around the origin
(the circle of a small radius around the point $\lambda=\infty$,
clockwise oriented around it).

We also recall that the Fredholm determinant
of a  linear integral operator is defined as
\begin{equation}\label{Fredholm}
\Det\left(1-\hat K\right)=1+\sum_{n=1}^\infty \frac{(-1)^n}{n!}
\int \det_{1\leq i,j\leq n}\left[ K(\lambda_i,\lambda_j)\right]\, \rmd^n\lambda.
\end{equation}
Using that $\Det(1-\hat K)=\exp\{\Tr\log(1-\hat K)\}$, and defining the function
$\log(1-\hat K)$ by its power series expansion in powers of $\hat K$,
we can also define the Fredholm determinant of the integral operator $\hat K$ by
the formula
\begin{equation}\label{detexp}
\Det\left(1-\hat K\right)=\exp\left\{-\sum_{n=1}^\infty \frac{1}{n}
\int K(\lambda_1,\lambda_2)\cdots K(\lambda_n,\lambda_1)\, \rmd^n\lambda\right\}.
\end{equation}
We now formulate the Fredholm determinant representation for the EFP,
derived in \cite{P-13}.

\begin{proposition}\label{prop1}
The EFP admits representation in terms of the Fredholm determinant
\begin{equation}\label{EFPasFdet}
F_{r,s,q}=\Det\left(1-\hat K\right),
\end{equation}
where $\hat K=\hat K_{r,s,q}$ is the integral operator acting
on the contour $C_0$ and possessing the kernel
\begin{equation}\label{Klamu}
K(\lambda,\mu)=\frac{e_{+}(\lambda)e_{-}(\mu)-e_{-}(\lambda)e_{+}(\mu)}
{2\pi\rmi(\lambda-\mu)},
\qquad
\lambda,\mu\in C_0,
\end{equation}
where the functions $e_\pm(\lambda)$, which depend on the integers
$r$, $s$, and $q$, i.e., $e_\pm(\lambda)=e_{\pm,r,s,q}(\lambda)$, are
\begin{equation}\label{em}
e_{-}(\lambda)=\frac{(\lambda-\alpha)^{r/2}}{(\lambda-1)^{(r+q)/2}\lambda^{s/2}}
\end{equation}
and
\begin{equation}\label{ep}
e_{+}(\lambda)=e_{-}(\lambda)E(\lambda),
\end{equation}
where the function $E(\lambda)$ is
\begin{equation}\label{Ela}
E(\lambda)=\oint_{C_\infty}\frac{(\nu-1)^{r+q}\nu^s}{(\nu-\alpha)^r(\nu-\lambda)}
\frac{\rmd \nu}{2\pi\rmi}
\end{equation}
and $C_0$ and $C_\infty$ are counter-clockwise oriented circles of small and large radii,
respectively.
\end{proposition}

The details of the proof can be found in \cite{P-13}. Here we just mention the idea
of the proof: upon change of the variable  $\nu\mapsto 1/\nu$, which replaces
the contour $C_\infty$ by the contour $C_0$, the trace of the $n$th power
of the operator $\hat K$ reads
\begin{multline}
\Tr\left(\hat K^n\right)=(-1)^{nq}
\oint_{C_0}\cdots\oint_{C_0}
\prod_{i=1}^n
\frac{(\alpha-\lambda_i)^{r} (1-\nu_i)^{r+q}}{(1-\lambda_i)^{r+q}(1-\alpha \nu_i)^{r}}
\\ \times
\frac{1}{\lambda_{i+1}^{s}\nu_i^{s+q}(1-\lambda_i \nu_i)(1-\lambda_{i+1}\nu_i)}
\frac{\rmd^n \nu\, \rmd^p \lambda}{(2\pi\rmi)^{2n}}
\end{multline}
where $\lambda_{p+1}:=\lambda_1$. Since all $\lambda_i$'s and $\nu_i$ are
integrated over the contours of small radii, the factors
$(1-\lambda_i \nu_i)^{-1}$ and $(1-\lambda_{i+1}\nu_i)^{-1}$ can be expanded
in the Taylor series. Re-ordering the summations and integrations here, various
Fredholm-type determinant representations can be obtained, in particular,
in terms of a finite-size matrix, which leads to the Hankel determinant
representation \eqref{newEFP}.

It is necessary to mention that
the representation \eqref{newEFP}, up to a redefinition of the discrete parameters
and when rewritten in its random matrix model integral form
(with a discrete measure), appears to be coinciding with certain quantity
discussed in a different context by Johansson \cite{J-00} (see Proposition 1.3 therein).
In turn, Baik and Rains \cite{BR-01,B-03} showed that this quantity can be
represented as certain average over circular unitary ensemble.
Using the Okamoto theory \cite{O-87}, Forrester and Witte \cite{FW-04} next showed
this average serves as $\tau$-function corresponding to
a classical solution of P6. In what follows we show that the this result,
which can be summarized in the form of Theorem \ref{EFPasP6},
can be established directly from
the Fredholm determinant representation \eqref{EFPasFdet}, using the notion
of the integrable integral operators introduced by Its, Izergin, Korepin, and Slavnov
\cite{IIKS-90}.

\subsection{The resolvent operator}

Consider Fredholm integral equations
\begin{equation}
f_{\pm}(\lambda)-\oint_{C_0} K(\lambda,\mu)f_\pm(\mu)\,\rmd \mu= e_\pm(\lambda),
\end{equation}
which can be written in operator form as
\begin{equation}\label{fKe}
\left[\left(1-\hat K\right)f_\pm\right](\lambda)=e_\pm(\lambda).
\end{equation}
For brevity we omit dependence of all functions on the parameter $\alpha$ and
discrete parameters $r$, $s$, and $q$.
We consider solutions of \eqref{fKe} belonging
to the class of analytic functions in $\lambda$-plane
which can be presented in a vicinity of the
point $\lambda=0$ in the form
\begin{equation}
f_\pm(\lambda)=\frac{h_1(\lambda)}{e_{-}(\lambda)}+h_2(\lambda)e_{-}(\lambda),
\end{equation}
where $h_1(\lambda)$ and $h_2(\lambda)$ are analytic functions at the origin.
For given values of $r$, $s$, and $q$  there exists
a finite number of points $\alpha=\alpha_1,\ldots,\alpha_p$, where $p=r(s+q)$, such
that $\dim\Ker(1-\hat K)=0$ for $\alpha\in\mathbb{C}\setminus \{\alpha_1,\ldots,\alpha_p\}$,
and $\dim\Ker(1-\hat K)=1$ at $\alpha\in\{\alpha_1,\ldots,\alpha_p\}$.
It is important to note that $\alpha_1,\ldots,\alpha_p\notin (0,1)$
that is the operator $1-\hat K$ is invertible for the case of our interest,
$\alpha\in(0,1)$. Therefore,
one can define the resolvent $\hat R$ of the operator $\hat K$,
\begin{equation}\label{KR}
\left(1-\hat K\right)\left(1+\hat R\right)=1.
\end{equation}
As we show below,
the kernel the operator $\hat R$ can be written
in terms of the functions $f_\pm(\lambda)$, exactly in the same way as
the kernel of the operator $\hat K$ is given in terms of the functions
$e_\pm(\lambda)$ \cite{IIKS-90}.

It is easy to see that \eqref{fKe} can be rewritten in the form
\begin{equation}\label{Xfe}
X(\lambda)\,\vec e(\lambda) = \vec f(\lambda),
\end{equation}
where
\begin{equation}
\vec e(\lambda)=
\begin{pmatrix}
e_{+}(\lambda) \\ e_{-}(\lambda)
\end{pmatrix},
\qquad
\vec f(\lambda)=
\begin{pmatrix}
f_{+}(\lambda) \\ f_{-}(\lambda)
\end{pmatrix},
\end{equation}
and
\begin{equation}\label{Xla}
X(\lambda)=
\begin{pmatrix}
1-H_{-+}(\lambda) &  H_{++}(\lambda) \\
-H_{--}(\lambda) & 1 + H_{+-}(\lambda)
\end{pmatrix},
\end{equation}
where the functions $H_{mn}(\lambda)$ are
\begin{equation}\label{Hmn}
H_{mn}(\lambda)=\oint_{C_0} \frac{e_m(\mu)f_n(\mu)}{\mu-\lambda}\,
\frac{\rmd \mu}{2\pi\rmi},
\qquad
m,n=\{+,-\}.
\end{equation}
It turns out that the matrix $X(\lambda)$ can be easily inverted, due to the following
remarkable property.
\begin{proposition}
The determinant of the matrix \eqref{Xla} is equal to one,
\begin{equation}
\det X(\lambda)=1,\qquad \lambda\in \mathbb{C}.
\end{equation}
\begin{proof}
To prove the result we use the same trick as in
\cite{KBI-93}, Chap.~XIV, App.~C, but we
will rely only the known structure of the kernel of
the operator $\hat K$, instead of that for the resolvent.
Note that the
functions $H_{+-}(\lambda)$ and $H_{-+}(\lambda)$
can be written in the form
\begin{equation}
H_{+-}(\lambda)=\oint_{C_0} \frac{\left[\left(1-\hat K\right) f_{+}\right](\mu)f_{-}(\mu)}{\mu-\lambda}
\frac{\rmd \mu}{2\pi\rmi}
\end{equation}
and
\begin{equation}
H_{-+}(\lambda)=\oint_{C_0} \frac{f_{+}(\nu)
\left[\left(1-\hat K\right) f_{-}\right](\nu)}{\nu-\lambda}\frac{\rmd \nu}{2\pi\rmi},
\end{equation}
respectively.
Subtracting the second expression from the first one and using the fact that
the kernel is symmetric, $K(\mu,\nu)=K(\nu,\mu)$, we have
\begin{equation}
H_{+-}(\lambda)-H_{-+}(\lambda)
=\oint_{C_0} \frac{\rmd \mu}{2\pi\rmi}  \oint_{C_0}
f_{-}(\mu) K(\mu,\nu) f_{+}(\nu)
\left(\frac{1}{\nu-\lambda}-\frac{1}{\mu-\lambda}\right) \rmd \nu.
\end{equation}
Since
\begin{equation}
\frac{1}{\nu-\lambda}-\frac{1}{\mu-\lambda}=
\frac{\mu-\nu}{(\mu-\lambda)(\nu-\lambda)}
\end{equation}
and
\begin{equation}
(\mu-\nu)K(\mu,\nu)=e_{+}(\mu)e_{-}(\nu)-e_{-}(\mu)e_{+}(\nu),
\end{equation}
it is fairly easy to see that in fact we have the relation
\begin{equation}
H_{+-}(\lambda)-H_{-+}(\lambda)
=H_{-+}(\lambda)H_{+-}(\lambda)-H_{--}(\lambda)H_{++}(\lambda).
\end{equation}
Obviously, the last relation implies that $\det X(\lambda)=1$.
\end{proof}
\end{proposition}

\begin{corollary}
The inverse matrix has the form
\begin{equation}\label{Xinv}
X^{-1}(\lambda)=
\begin{pmatrix}
1+H_{+-}(\lambda) &  -H_{++}(\lambda) \\
H_{--}(\lambda) & 1 - H_{-+}(\lambda)
\end{pmatrix}.
\end{equation}
\end{corollary}

\begin{corollary}
The functions $e_\pm(\lambda)$ solve the
following integral equations
\begin{equation}\label{fRelong}
f_{\pm}(\lambda)= e_\pm(\lambda)+\oint_{C_0} R(\lambda,\mu)e_\pm(\mu)\,\rmd \mu,
\end{equation}
where
\begin{equation}\label{kernelR}
R(\lambda,\mu)=
\frac{f_{+}(\lambda)f_{-}(\mu)-f_{-}(\lambda)f_{+}(\mu)}
{2\pi\rmi(\lambda-\mu)}.
\end{equation}
\end{corollary}
\begin{proof}
Using \eqref{Hmn} it is straightforward to rewrite the formula
\begin{equation}
X^{-1}(\lambda)\, \vec f (\lambda)=\vec e(\lambda)
\end{equation}
in the form of the integral equations
\eqref{fRelong} with the kernel \eqref{kernelR}.
\end{proof}

Note that \eqref{fRelong} can be written in the following operator form
\begin{equation}\label{fRe}
f_\pm(\lambda)=\left[\left(1+\hat R\right)e_\pm\right](\lambda).
\end{equation}
Comparing \eqref{fRe} with \eqref{fKe}, we conclude that the operator
$\hat R$ is the resolvent of the operator $\hat K$, defined by the formula
\eqref{KR}.
In other words, we have just showed that if the operator
$\hat K$ possesses the kernel \eqref{Klamu}, then
the functions $f_\pm(\lambda)$, defined as the solutions
of the equations \eqref{fKe}, determine the kernel
of the resolvent by formula \eqref{kernelR}.

In a similar manner one can calculate
action of the operator $(1-\hat K)^{-1}=1+\hat R$ on various
functions. In what follows we often use the following result.
\begin{lemma}
Let the vector $g(\lambda;\nu)$ to be
\begin{equation}\label{gpm}
\vec g(\lambda;\nu)=
\begin{pmatrix}
g_{+}(\lambda;\nu) \\ g_{-}(\lambda;\nu)
\end{pmatrix},
\qquad
g_\pm (\lambda;\nu)= \frac{e_\pm(\lambda)}{\lambda-\nu},
\end{equation}
then
\begin{equation}\label{gvecf}
\left[\left(1+\hat R\right)\vec g(\cdot;\nu)\right] (\lambda)
= \frac{1}{\lambda-\nu}X^{-1}(\nu) \vec f(\lambda).
\end{equation}
\end{lemma}
\begin{proof}
A direct calculation gives us
\begin{align}\label{Rgpm}
\left[\left(1+\hat R\right)g_\pm(\cdot;\nu)\right] (\lambda)
&
=\frac{e_\pm(\lambda)}{\lambda-\nu}
+\oint_{C_0}R(\lambda,\mu)  \frac{e_\pm(\mu)}{\mu-\nu}\,\rmd\mu
\notag\\  &
=\frac{e_\pm(\lambda)}{\lambda-\nu}
+\frac{1}{\lambda-\nu}\oint_{C_0}R(\lambda,\mu) e_\pm(\mu)\left(1+ \frac{\lambda-\mu}{\mu-\nu}
\right)\rmd\mu
\notag\\  &
=\frac{f_\pm(\lambda)}{\lambda-\nu}
+\frac{1}{\lambda-\nu}\oint_{C_0}R(\lambda,\mu) (\lambda-\mu)
e_\pm(\mu)\rmd\mu
\notag\\  &
=\frac{f_\pm(\lambda)}{\lambda-\nu}
+H_{\pm-}(\nu)\frac{f_{+}(\lambda)}{\lambda-\nu}
-H_{\pm+}(\nu)\frac{f_{-}(\lambda)}{\lambda-\nu},
\end{align}
where relations \eqref{fRe} and \eqref{kernelR} have been used.
Comparing the result in \eqref{Rgpm} with \eqref{Xinv}, we arrive at
\eqref{gvecf}.
\end{proof}

\subsection{The Fuchs pair}

One of the main properties of the functions $f_\pm(\lambda)$ is that that obey
linear first-order differential equations
in the variables $\alpha$ and $\lambda$.  These equations appear to be a vector
form of the famous Fuchs pair for P6.

Note that there exist also difference equations with respect to the discrete parameters
$r$, $s$, and $q$, shifting them by $\pm 1$. These equations, together with
the one with respect to $\alpha$, constitute the set of the Lax-type equations
of our problem. They can be viewed as the so-called Schlesinger transformations
for the Fuchs pair given below, or derived in a straightforward way using the
method presented in the proof.

\begin{proposition}\label{Fpair}
The functions $f_\pm(\lambda)$ is the solution of the system of linear
differential equations
\begin{equation}
\frac{\rmd}{\rmd \lambda}\vec f(\lambda)
=A(\lambda)\vec f(\lambda),
\qquad
\frac{\rmd}{\rmd \alpha}\vec f(\lambda)
=B(\lambda)\vec f(\lambda),
\end{equation}
with
\begin{equation}\label{AB}
A(\lambda)=
\frac{A_0}{\lambda}+\frac{A_1}{\lambda-1}+\frac{A_\alpha}{\lambda-\alpha},
\qquad
B(\lambda)=-\frac{A_\alpha}{\lambda-\alpha},
\end{equation}
where the matrices $A_0$, $A_1$, and $A_\alpha$ are given by
\begin{equation}\label{Anu}
A_\nu=\frac{\theta_\nu}{2}\, X(\nu) S(\nu) X^{-1}(\nu)
\qquad (\nu=0,1,\alpha).
\end{equation}
Here
\begin{equation}\label{eq:def-theta-finite}
\theta_0=s,\qquad
\theta_1=r+q,\qquad
\theta_\alpha=-r,
\end{equation}
and the matrix $X(\lambda)$ is given in
\eqref{Xla} and the matrix $S(\lambda)$ is
\begin{equation}\label{Sla}
S(\lambda)
=
\begin{pmatrix}
1 & -2E(\lambda) \\ 0 & -1
\end{pmatrix}.
\end{equation}
\end{proposition}
\begin{proof}
Here we consider a proof based on the straightforward differentiation of the
functions $f_\pm(\lambda)$; a simpler proof,
based on analytical properties of the matrix
$X(\lambda)$, is given in Sect. \ref{Sec:tau-f}.

We first consider the derivatives of functions $f_\pm(\lambda)$ with
respect to the parameter $\alpha$.
We begin with calculating the derivatives of functions $e_\pm(\lambda)$.
{}From \eqref{em}, for the function $e_{-}(\lambda)$ we have
\begin{equation}\label{daem}
\frac{\rmd}{\rmd\alpha}e_{-}(\lambda)=-\frac{r}{2(\lambda-\alpha)} e_{-}(\lambda).
\end{equation}
To obtain the similar relation for the function $e_{+}(\lambda)$, let
us consider first the function $E(\lambda)$, see \eqref{ep}, for which
we have
\begin{equation}
\frac{\rmd}{\rmd\alpha}E(\lambda)=
r\oint_{C_\infty}\frac{(\nu-1)^{r+q}\nu^s}{(\nu-\alpha)^{r+1}(\nu-\lambda)}
\frac{\rmd \nu}{2\pi\rmi}.
\end{equation}
Using
\begin{equation}
\frac{1}{(\nu-\alpha)(\nu-\lambda)}=
\frac{1}{\lambda-\alpha}\left(\frac{1}{\nu-\lambda}
-\frac{1}{\nu-\alpha}\right),
\end{equation}
we can rewrite the relation above as
\begin{equation}
\frac{\rmd}{\rmd\alpha}E(\lambda)
=\frac{r}{\lambda-\alpha}\left(E(\lambda)-E(\alpha)\right)
\end{equation}
and therefore
\begin{equation}\label{daep}
\frac{\rmd}{\rmd\alpha}e_{+}(\lambda)=
\frac{r}{2(\lambda-\alpha)} e_{+}(\lambda)
-\frac{rE(\alpha)}{\lambda-\alpha}e_{-}(\lambda),
\end{equation}
where we have also used \eqref{daem}.

To obtain derivatives of the functions $f_\pm(\lambda)$ with respect to $\alpha$,
we differentiate the defining relation \eqref{fKe}, that gives
\begin{multline}\label{diffeq}
\frac{\rmd}{\rmd\alpha}  f_{\pm}(\lambda)-
\oint_{C_0}\left[ \frac{\rmd}{\rmd\alpha} K(\lambda,\mu)\right]f_\pm(\mu)\,\rmd \mu
-\oint_{C_0}K(\lambda,\mu) \frac{\rmd}{\rmd\alpha} f_\pm(\mu)\,\rmd \mu
\\
=\frac{\rmd}{\rmd\alpha} e_\pm(\lambda).
\end{multline}
Using \eqref{daem} and \eqref{daep}, the derivative of
$K(\lambda,\mu)$ can be computed, with the result
\begin{equation}\label{daK}
\frac{\rmd}{\rmd\alpha}
K(\lambda,\mu)
= -r \frac{e_{+}(\lambda)e_{-}(\mu) +e_{-}(\lambda)e_{+}(\mu)
-2E(\alpha)e_{-}(\lambda)e_{-}(\mu)}{4\pi\rmi(\lambda-\alpha)(\mu-\alpha)}.
\end{equation}
Thus, moving the second term in the left-hand side of \eqref{diffeq}
to the right, we can rewrite \eqref{diffeq} as follows:
\begin{multline}
\left[\left(1-\hat K\right)\frac{\rmd}{\rmd\alpha}f_\pm\right](\lambda)
= \frac{\rmd}{\rmd\alpha} e_\pm(\lambda)
\\
- \frac{r}{2(\lambda-\alpha)}
\big[H_{-\pm}(\alpha) e_{+}(\lambda) +H_{+\pm}(\alpha) e_{-}(\lambda)
-2 E(\alpha) H_{-\pm}(\alpha) e_{-}(\lambda)\big].
\end{multline}
In the case of the plus-sign subscript, using \eqref{daep}, we thus obtain
\begin{multline}\label{withplus}
\left[\left(1-\hat K\right)\frac{\rmd}{\rmd\alpha}f_{+}\right](\lambda)
=  \big(1-H_{-+}(\alpha)\big) \frac{r e_{+}(\lambda)}{2(\lambda-\alpha)}
\\
-\left[H_{++}(\alpha)+2 E(\alpha) \big(1- H_{-+}(\alpha)\big)\right]
\frac{r e_{-}(\lambda)}{2(\lambda-\alpha)}.
\end{multline}
Similarly, in the case of the minus-sign subscript, using \eqref{daem}, we get
\begin{multline}\label{withminus}
\left[\left(1-\hat K\right)\frac{\rmd}{\rmd\alpha}f_{-}\right](\lambda)
=  - H_{--}(\alpha)\frac{re_{+}(\lambda)}{2(\lambda-\alpha)}
\\
-\big[1+H_{+-}(\alpha)+2 E(\alpha) H_{--}(\alpha)\big]
\frac{re_{-}(\lambda)}{2(\lambda-\alpha)}.
\end{multline}
Now the desired result can be obtained by acting with the operator
$(1-\hat K)^{-1}=1+\hat R$ on these two relations.

Comparing \eqref{withplus} and \eqref{withminus} with \eqref{Xla}, one can note that
these two relations can also be written as
\begin{equation}
\left[\left(1-\hat K\right)\frac{\rmd}{\rmd\alpha}\vec f\right](\lambda)
=\frac{r}{2}\, X(\alpha) S(\alpha)
\vec g(\lambda;\alpha),
\end{equation}
where the matrix $S(\lambda)$ is defined in \eqref{Sla}.
Specifying $\nu=\alpha$ in \eqref{gvecf}, we obtain
\begin{equation}\label{daf}
\frac{\rmd}{\rmd\alpha}\vec f(\lambda) =
\frac{r}{2(\lambda-\alpha)} X(\alpha) S(\alpha) X^{-1}(\alpha)\, \vec f(\lambda).
\end{equation}
Clearly, the matrix standing in the right-hand in \eqref{daf}
side is exactly the matrix $B(\lambda)$ appearing in \eqref{AB}.

Let us now find the derivatives of the functions $f_\pm(\lambda)$
with respect to the variable $\lambda$. We first obtain
those of the functions $e_\pm(\lambda)$. It is
straightforward for $e_{-}(\lambda)$,
\begin{equation}\label{dlem}
\frac{\rmd}{\rmd\lambda}e_{-}(\lambda)= -\frac{1}{2}
\left(\frac{s}{\lambda}
+\frac{r+q}{\lambda-1}-\frac{r}{\lambda-\alpha}\right)e_{-}(\lambda).
\end{equation}
To obtain that of $e_{+}(\lambda)$, we note that
the derivative of the function $E(\lambda)$, after integration by parts, reads
\begin{align}
\frac{\rmd}{\rmd\lambda}E(\lambda)
&
=\oint_{C_\infty}\frac{\nu^s(\nu-1)^{r+q}}{(\nu-\alpha)^r(\nu-\lambda)^2}
\frac{\rmd \nu}{2\pi\rmi}
\notag\\ &
=\oint_{C_\infty}\left(\frac{\rmd}{\rmd\nu}\frac{\nu^s(\nu-1)^{r+q}}{(\nu-\alpha)^r}
\right)\frac{1}{\nu-\lambda}
\frac{\rmd \nu}{2\pi\rmi}
\notag\\ &
=\frac{s}{\lambda}(E(\lambda)-E(0))+\frac{r+q}{\lambda-1}(E(\lambda)-E(1))
-\frac{r}{\lambda-\alpha} (E(\lambda)-E(\alpha)).
\end{align}
Hence,
\begin{multline}\label{dlep}
\frac{\rmd}{\rmd\lambda}e_{+}(\lambda)=
\frac{1}{2}\left(\frac{s}{\lambda}+\frac{r+q}{\lambda-1}
-\frac{r}{\lambda-\alpha}\right) e_{+}(\lambda)
\\
-\left(\frac{sE(0)}{\lambda}+\frac{(r+q)E(1)}{\lambda-1}
-\frac{rE(\alpha)}{\lambda-\alpha}\right) e_{-}(\lambda),
\end{multline}
where \eqref{dlem} have been used.

Now we are prepared for calculations with the functions $f_\pm(\lambda)$.
Differentiate the defining relation \eqref{fKe} to obtain
\begin{multline}\label{diffeq2}
\frac{\rmd}{\rmd\lambda}f_\pm(\lambda)
-\oint_{C_0}
\left[\left(\frac{\rmd}{\rmd\lambda}+\frac{\rmd}{\rmd\mu}\right)K(\lambda,\mu) \right]
f_\pm(\mu)\,
\rmd \mu
-\oint_{C_0}K(\lambda,\mu) \frac{\rmd}{\rmd\mu}f_\pm(\mu)\,
\rmd \mu
\\
=\frac{\rmd}{\rmd\lambda}e_\pm(\lambda).
\end{multline}
Here the second term can be expressed in terms of the functions $e_\pm(\lambda)$.
Indeed, since
\begin{equation}
\left(\frac{\rmd}{\rmd\lambda}+\frac{\rmd}{\rmd\mu}\right)\frac{1}{\lambda-\mu}=0,
\end{equation}
differentiation of the kernel affects only the numerator,
\begin{equation}
\left(\frac{\rmd}{\rmd\lambda}+\frac{\rmd}{\rmd\mu}\right)K(\lambda,\mu)
=\frac{e_{+}'(\lambda)e_{-}(\mu)-e_{-}'(\lambda)e_{+}(\mu)
+e_{+}(\lambda)e_{-}'(\mu)-e_{-}(\lambda)e_{+}'(\mu)}{2\pi\rmi(\lambda-\mu)},
\end{equation}
where prime denotes the derivative with respect to the argument of the function.
Using \eqref{dlem} and \eqref{dlep}, we obtain
\begin{align}
&
\left(\frac{\rmd}{\rmd\lambda}+\frac{\rmd}{\rmd\mu}\right)K(\lambda,\mu)
\notag\\ & \qquad
=-\frac{1}{2}\left[\frac{s}{\lambda\mu}
+\frac{r+q}{(\lambda-1)(\mu-1)}-\frac{r}{(\lambda-\alpha)(\mu-\alpha)}
\right]\frac{e_{+}(\lambda)e_{-}(\mu)+e_{-}(\lambda)e_{+}(\mu)}{2\pi\rmi}
\notag\\ & \qquad\qquad
+\left[
\frac{sE(0)}{\lambda\mu}
+\frac{(r+q)E(1)}{(\lambda-1)(\mu-1)}
-\frac{rE(\alpha)}{(\lambda-\alpha)(\mu-\alpha)}
\right]\frac{e_{-}(\lambda)e_{-}(\mu)}{2\pi\rmi}.
\end{align}
Thus, moving the second term in the left-hand side of \eqref{diffeq2} to the right,
we arrive at the following relation:
\begin{multline}\label{longuy}
\left[\left(1-\hat K\right)f_\pm'\right](\lambda)
\\
=\frac{\rmd}{\rmd\lambda}e_\pm(\lambda)
-\frac{s}{2\lambda}
\big[H_{-\pm}(0)e_{+}(\lambda)+\big(H_{+\pm}(0)-2E(0)H_{-\pm}(0)\big)
e_{-}(\lambda)\big]
\\
-\frac{r+q}{2(\lambda-1)}
\big[H_{-\pm}(1)e_{+}(\lambda)+\big(H_{+\pm}(1)-2E(1)H_{-\pm}(1)\big)
e_{-}(\lambda)\big]
\\
+\frac{r}{2(\lambda-\alpha)}
\big[H_{-\pm}(\alpha)e_{+}(\lambda)+\big(H_{+\pm}(\alpha)-2E(\alpha)H_{-\pm}(\alpha)\big)
e_{-}(\lambda)\big].
\end{multline}
Using \eqref{dlem} and \eqref{dlep} and switching to the vector form, one can see that
\eqref{longuy} is just the relation
\begin{multline}
\left[\left(1-\hat K\right)\vec{f'}\right](\lambda)
\\
=\frac{s}{2}\,X(0)S(0)\vec g(\lambda;0)
+\frac{r+q}{2}\,X(1)S(1)\vec g(\lambda;1)
-\frac{r}{2}\,X(\alpha)S(\alpha)\vec g(\lambda;\alpha),
\end{multline}
where functions $g_\pm(\lambda;\nu)$ are defined in \eqref{gpm}.

Finally, acting with the operator $(1-\hat K)^{-1}=1+\hat R$ on the last relation
by making use of \eqref{Rgpm} we obtain
\begin{multline}\label{dlf}
\frac{\rmd}{\rmd\lambda}\vec f(\lambda)
=
\bigg\{\frac{s}{2\lambda} X(0) S(0) X^{-1}(0)+
\frac{r+q}{2(\lambda-1)} X(1) S(1) X^{-1}(1)
\\
-\frac{r}{2(\lambda-\alpha)} X(\alpha) S(\alpha) X^{-1}(\alpha)
\bigg\}\, \vec e(\lambda).
\end{multline}
Clearly, the matrix appearing in the right-hand side of \eqref{dlf} coincides with
the matrix $A(\lambda)$ given in \eqref{AB}.
\end{proof}

\section{EFP and P6}\label{Sec:3}

The main goal of this section is to prove Theorem \ref{EFPasP6}.
We begin with recalling some basic facts of the theory of P6.
Next, we discuss analytic properties of the matrix $X(\lambda)$
in terms of the related objects of the theory of integrable systems. Finally,
we use all these ingredients to
show that the Fredholm determinant
of the operator $\hat K$ serves as a $\tau$-function of P6.

\subsection{The sixth Painlev\'e equation}\label{Sec:tau-f}

The canonical form of P6 (cf.~\cite{I-44}) reads,
\begin{multline}\label{P6}
\frac{\rmd^2 y}{\rmd\alpha^2}
=\frac{1}{2}\left(\frac{1}{y}+\frac{1}{y-1}+\frac{1}
{y-\alpha}\right)\left(\frac{\rmd y}{\rmd\alpha}\right)^2
-\left(\frac{1}{\alpha}+\frac{1}{\alpha-1}+\frac{1}{y-\alpha}\right)\frac{\rmd y}
{\rmd\alpha}
\\
+\frac{y(y-1)(y-\alpha)}{\alpha^2(\alpha^2-1)}
\left(a+b\frac{\alpha}{y^2}+c\frac{\alpha-1}{(y-1)^2}
+d\frac{\alpha(\alpha-1)}{(y-\alpha)^2}\right).
\end{multline}
For generic choice of the coefficients $a$, $b$, $c$, and $d$ solutions $y=y(\alpha)$
are transcendental functions called the sixth Painlev\'e transcendents. These functions
cannot be expressed by a finite number of operations through the known
transcendental (hypergeometric-type, elliptic, etc), algebraic, and elementary
functions (for a more precise formulation, see \cite{U-95,U-06}).

At the same time, it is known (cf.~\cite{O-87}) that for some particular choices of
coefficients there exist special, and even general, solutions which can be expressed
in terms of Gauss hypergeometric, algebraic, elementary functions
(and composition of elliptic and hypergeometric functions).
In Sect.~\ref{Sec:FDR} we show that the EFP can be expressed in terms
of a rational solution of P6.
Classification of rational solutions of P6 via its representation
as the isomonodromy deformation of a linear $2\times2$ matrix Fuchsian ODE
of the first order is given in~\cite{M-01}. We do not use
this classification result, however,
we need the above mentioned isomonodromy representation of P6.

Consider the $2\times 2$ matrix linear Fuchsian ODE
\begin{equation}\label{eqlambda}
\frac{\rmd}{\rmd\lambda}Y(\lambda)=A(\lambda)Y(\lambda),\qquad
A(\lambda)=\frac{A_0}{\lambda}
+\frac{A_1}{\lambda-1}+\frac{A_\alpha}{\lambda-\alpha},
\end{equation}
where matrices $A_0$, $A_1$, and $A_\alpha$ are independent
of $\lambda$.
In generic situation the isomonodromy deformations coincide with
the Schlesinger deformations, and they are
governed  by the following equation with respect to $\alpha$:
\begin{equation}\label{eqalpha}
\frac{\rmd}{\rmd\alpha}Y(\lambda)=B(\lambda)Y(\lambda),\qquad
B(\lambda)=-\frac{A_\alpha}{\lambda-\alpha}.
\end{equation}
The compatibility condition of \eqref{eqlambda} and \eqref{eqalpha},
which is called the Schlesinger system, is equivalent,
according to Jimbo and Miwa \cite{JM-81}, to P6 equation \eqref{P6}.
More precisely, choose normalization of \eqref{eqlambda} such that
the matrices $A_0$, $A_1$, and $A_\alpha$ are traceless, and
the matrix $A_\infty:=-(A_0+A_1+A_\alpha)$ is diagonal
\begin{equation}
A_\infty
=\begin{pmatrix}
\theta_\infty/2 & 0 \\ 0 & -\theta_\infty/2
\end{pmatrix}, \qquad \theta_\infty\ne 0.
\end{equation}
Denote the eigenvalues of the matrices $A_k$ as $\pm\theta_k/2$, $k=0,1,\alpha,\infty$.
Actually, while matrices $A_k$ are functions of $\alpha$, their
eigenvalues $\theta_k$ are in fact independent
of $\alpha$, since they are the first integrals of the Schlesinger system.

Denoting by $(A_k)_{ij}$, $i,j=1,2$, the matrix elements of $A_k$, and noting the relation
\begin{equation}
(A_0)_{12}+(A_1)_{12}+(A_\alpha)_{12}=0,
\end{equation}
one can verify that the equation
\begin{equation}
\frac{(A_0)_{12}}{y}+\frac{(A_1)_{12}}{y-1}+\frac{(A_\alpha)_{12}}{y-\alpha}=0
\end{equation}
has in the generic situation (that is, for $(A_1)_{ij}+\alpha (A_\alpha)_{ij}\ne 0$)
unique solution $y$, provided that the coefficients in \eqref{P6} are chosen as follows:
\begin{equation}
a=\frac{(\theta_\infty-1)^2}{2},\qquad
b=-\frac{\theta_0^2}{2},\qquad
c=\frac{\theta_1^2}{2},\qquad
d=\frac{1-\theta_\alpha^2}{2}.
\end{equation}
One can associate with isomodromy deformations of \eqref{eqlambda}
the so-called $\tau$-function, which is defined modulo a constant
factor by the relation
\begin{align}\label{eq:def-tau-JM}
\frac{\rmd}{\rmd\alpha}\log \tau
&\equiv\underset{\lambda=\alpha}\res\tr\left(
\frac{A_0+\theta_0/2}{\lambda}+\frac{A_1+\theta_1/2}{\lambda-1}
+\frac{A_\alpha+\theta_\alpha/2}{\lambda-\alpha}
\right)^2\\
&=\tr\left[A_\alpha\left(\frac{A_0}{\alpha}+\frac{A_1}{\alpha-1}\right)\right]
+\frac{\theta_\alpha}{2}\left(\frac{\theta_0}{\alpha}+\frac{\theta_1}{\alpha-1}\right).
\nonumber
\end{align}
We note that our defining relation for the $\tau$-function~\eqref{eq:def-tau-JM}
(with the right-hand side being the Hamiltonian of P6) exactly reproduce
that of Jimbo and Miwa \cite{JM-81}, since we use the traceless normalization
of the matrices $A_k$, $k=0,1,\alpha$.

The function
\begin{multline}
\sigma=\alpha(\alpha-1)\frac{\rmd}{\rmd\alpha}
\log\tau
\\
+(\nu_1\nu_2+\nu_1\nu_3+\nu_2\nu_3) \alpha-
\frac{\nu_1\nu_2+\nu_1\nu_3+\nu_2\nu_3+\nu_1\nu_4+\nu_2\nu_4+\nu_3\nu_4}{2},
\end{multline}
where
\begin{equation}\label{nuks}
\nu_1=\frac{\theta_\alpha+\theta_\infty}{2},\quad
\nu_2=\frac{\theta_\alpha-\theta_\infty}{2},\quad
\nu_3=-\frac{\theta_0+\theta_1}{2},\quad
\nu_4=-\frac{\theta_0-\theta_1}{2},
\end{equation}
satisfies the equation (the so-called $\sigma$-form of P6):
\begin{multline}\label{eqsigma}
\alpha^2(\alpha-1)^2\sigma'(\sigma'')^2
+\left\{(1-2\alpha)(\sigma')^2+2\sigma\sigma'+\nu_1\nu_2\nu_3\nu_4\right\}^2
\\
=\left(\sigma'+\nu_1^2\right)
\left(\sigma'+\nu_2^2\right)
\left(\sigma'+\nu_3^2\right)
\left(\sigma'+\nu_4^2\right).
\end{multline}
The prime denotes differentiation with respect to $\alpha$.
We note the $\sigma$-function can also be written in the form
\begin{equation}\label{sigmatr}
\sigma=\alpha(\alpha-1)\frac{\rmd}{\rmd\alpha}
\tr\left[A_\alpha\left(\frac{A_0}{\alpha}+\frac{A_1}{\alpha-1}\right)\right]
+\nu_1\nu_2 \alpha-\frac{\nu_1\nu_2+\nu_3\nu_4}{2},
\end{equation}
which will be convenient for our purposes below.

\subsection{Riemann-Hilbert problem}\label{sec:RH}

Our result stated in Proposition \ref{Fpair} yields a special vector solution
of the system of ODEs \eqref{eqlambda}, \eqref{eqalpha} in terms of the solution
of the integral equation \eqref{fKe}, where the matrices $A_0$, $A_1$, and $A_\alpha$  are
given by \eqref{Anu}--\eqref{Sla}, \eqref{Xla}, and \eqref{Hmn}.
Here we construct a matrix fundamental solution corresponding to the
vector one, calculate its monodromy data and, finally, reformulate the problem of solving
of the Fredholm integral equation as the solution of a proper Riemann-Hilbert
problem or as a
solution of the inverse monodromy problem for ODE \eqref{eqlambda}.
As a byproduct of our considerations here we give another proof of Proposition~\ref{Fpair}
and obtain a formula for $X'(\lambda)$ which will appear important in what follows.
It should be stressed that this formula, just like the statement  of
Proposition~\ref{Fpair}, can be proved by the direct differentiation, while
their derivations based on the analytic properties of $X(\lambda)$
are easier modulo the considerations given below.

First, we discuss the structure of the matrix $Y(\lambda)$ in our case.
Looking at \eqref{Anu} and taking into account that the matrix $S(\lambda)$ admits factorization
\begin{equation}\label{Sfactor}
S(\lambda)=
\begin{pmatrix}
1 & E(\lambda) \\ 0 & 1
\end{pmatrix}
\begin{pmatrix}
1 & 0 \\ 0 & -1
\end{pmatrix}
\begin{pmatrix}
1 & -E(\lambda) \\ 0 & 1
\end{pmatrix},
\end{equation}
we observe that the matrices $A_0$, $A_1$, and $A_\alpha$ are in fact residues of
the logarithmic derivative
$[\partial_\lambda Y(\lambda)]Y(\lambda)^{-1}$ of the matrix
\begin{equation}\label{matY}
Y(\lambda)=X(\lambda)
\begin{pmatrix}
1 & E(\lambda) \\ 0 & 1
\end{pmatrix}
\begin{pmatrix}
1/e_{-}(\lambda) &  0\\
0 &  e_{-}(\lambda)
\end{pmatrix},
\end{equation}
where the functions $E(\lambda)$ and $e_{-}(\lambda)$ are given by
\eqref{Ela} and \eqref{em}, respectively. Below we show that $Y(\lambda)$ is
a fundamental solution of the Fuchs pair \eqref{eqlambda}, \eqref{eqalpha}.
Note that the second column of $Y(\lambda)$ is nothing but our vector solution
$\vec f(\lambda)$.

The standard approach of proving that $Y(\lambda)$ is the solution of the
Fuchs pair~\eqref{eqlambda}, \eqref{eqalpha}
consists in inspecting its analytic properties.
Equation~\eqref{matY} defines the matrix $Y(\lambda)$ as analytic
function on $\mathbb C\setminus\{C_0\cup C_\infty\}$,
i.e., piecewise analytic in $\mathbb C$,
with singular points at $\lambda=0,\alpha, 1$, and $\infty$. In case some numbers
$r$, $s$, or $q$ are odd one has to
make one or two cuts  in $\mathbb C$ connecting the singular points.
If so, then one should consider $\mathbb C$ with the cuts,
that makes no substantial changes in the following.

The boundary values of the functions $X(\lambda)$ and $E(\lambda)$
suffer jumps on the oriented contours $C_0$ and $C_\infty$ (recall that
both contours are counter-clockwise oriented with respect to the origin).
These jumps can be found with the help of the Sokhotski-Plemelj formulas.
Namely, let $\Gamma$ be an oriented contour; denoting
$Y^{\pm}_{\Gamma}(\lambda)=Y(\lambda\pm\epsilon)$
as $\epsilon\to 0$ with $\lambda\in\Gamma$, such that the point
$\lambda+\epsilon$  (respectively, $\lambda-\epsilon$) lies in the domain on the left-hand
(right-hand) side of the oriented contour $\Gamma$, we have
\begin{equation}\label{YYG}
Y^{+}_{\Gamma}(\lambda)=Y^{-}_{\Gamma}(\lambda)G_{\Gamma}^{}, \qquad
\lambda \in \Gamma, \qquad \Gamma=\{C_0, C_\infty\},
\end{equation}
where
\begin{equation}\label{GG}
G_{C_0}=
\begin{pmatrix}
1& 0 \\ -1 & 1
\end{pmatrix},
\qquad
G_{C_\infty}=
\begin{pmatrix}
1& 1 \\ 0 & 1
\end{pmatrix}.
\end{equation}
Using equations \eqref{matY}, \eqref{Xla}, \eqref{Hmn}, and \eqref{Ela} one confirms
the following asymptotic behavior at the singular points:
\begin{multline}\label{eq:Y-asympt-finite}
Y(\lambda)\underset{\lambda\to\nu}
=\left\{\Phi_0(\nu)+\sum_{k=1}^\infty\Phi_k(\nu)(\lambda-\nu)^k\right\}
\begin{pmatrix}
(\lambda-\nu)^{\theta_\nu/2} & 0\\
0 & (\lambda-\nu)^{-\theta_\nu/2}
\end{pmatrix},
\\
\nu=0,\alpha,1,
\end{multline}
where $\det\Phi_0(\nu)=1$ and the normalization condition at the point $\lambda=\infty$ is
\begin{equation}\label{eq:Y-asympt-infty}
Y(\lambda)\underset{\lambda\to\infty}
=\left\{I+\sum_{k=1}^\infty\Phi_k(\infty)\left(\frac1{\lambda}\right)^k\right\}
\begin{pmatrix}
\lambda^{-\theta_\infty/2} & 0 \\
0 & \lambda^{\theta_\infty/2}
\end{pmatrix},
\quad\theta_\infty=-(s+q).
\end{equation}
The series in \eqref{eq:Y-asympt-finite} and \eqref{eq:Y-asympt-infty}
are convergent in some neighborhoods of the corresponding
singular points, and $\theta_\nu$'s are given by \eqref{eq:def-theta-finite}.

We can summarize the analytic properties of $Y(\lambda)$ as the following singular
Riemann-Hilbert problem:
\begin{enumerate}
\item
The $2\times2$ matrix function $Y(\lambda)$ is analytic
in $\mathbb C\setminus\{C_0\cup C_\infty\}$, where circles $C_0$ and $C_\infty$
are centered at $0$ and $\infty$, respectively;
\item
The boundary values on the contours $C_0$ and $C_\infty$ satisfy the
jump conditions ~\eqref{YYG}, \eqref{GG};
\item
The function $Y(\lambda)$ has singular points at $\lambda=0,\alpha,1,\infty$
where it has asymptotic expansions~\eqref{eq:Y-asympt-finite}
and \eqref{eq:Y-asympt-infty};
\item
The points $\lambda=\alpha$ and $\lambda=1$ belong to the annulus
domain between the circles $C_0$ and $C_\infty$.
\end{enumerate}

The existence of the solution of the Riemann-Hilbert problem can be established by
the reference to the integral equation~\eqref{fKe} and
representation \eqref{matY} for $Y(\lambda)$; the uniqueness can be proven
by standard arguments based on the Liouville theorem. In fact, the
Riemann-Hilbert problem is a reformulation of the analytic properties
of the matrix-function $Y(\lambda)$, which thus can be seen as its implicit solution.
Note that this Riemann-Hilbert problem can solved explicitly
with the help of the theory Schlesinger transformations.

The solution of the Riemann-Hilbert problem delivers
the generalized solution of the Fuchs
pair \eqref{eqlambda}, \eqref{eqalpha}.
This can be established in the standard way by noticing that the logarithmic
derivatives $\partial_\lambda Y(\lambda)\,Y(\lambda)^{-1}$ and
$\partial_\alpha Y(\lambda)\,Y(\lambda)^{-1}$ are rational functions on
$\mathbb{C}$ with the first order poles at $0$, $\alpha$, $1$, and $\infty$.
Calculation of the residues at these points with the help of representation
\eqref{matY} provides us an alternative proof of Proposition~\ref{Fpair},
since our vector solution
$\vec f(\lambda)$ is nothing but the second column of $Y(\lambda)$.
Note that, if one does not assume representation  \eqref{matY}, then
the residue-matrices $A_\nu$, $\nu=0,\alpha,1$, can be calculated in terms of the entries
the matrix $\Phi_1(\infty)$.

Using the solution of the Riemann-Hilbert problem we can construct
the classical solution to the pair  \eqref{eqlambda}, \eqref{eqalpha}
just by making analytic continuation on the whole complex plane of any
piece of the function $Y(\lambda)$; since we have three pieces
we obtain three fundamental solutions related to each other by
constant matrices.  We take the canonical solution, i.e., the one
obtained by the analytic continuation from the disk centered at infinity
with the asymptotics \eqref{eq:Y-asympt-infty}.
Our purpose is to characterize this solution by its monodromy data.
If we denote this solution as $Y_\infty(\lambda)$ then it can be
characterized by the following asymptotic expansions.
\begin{proposition}
The matrix-function $Y_\infty(\lambda)$
is the unique fundamental solution of
\eqref{eqlambda} with the residue matrices $A_\nu$
defined by \eqref{Anu} and \eqref{eq:def-theta-finite}
with the following expansions at the regular singular points:
\begin{multline}\label{Yinfty}
Y_\infty(\lambda)
\underset{\lambda\to\nu}=
\left\{\Phi_0(\nu)+\sum_{k=1}^\infty\Phi_k(\nu)(\lambda-\nu)^k\right\}
\begin{pmatrix}
(\lambda-\nu)^{\theta_\nu/2} & 0 \\
0 & (\lambda-\nu)^{-\theta_\nu/2}
\end{pmatrix}
\mathcal{C}_\nu,
\\
\qquad
\nu=0,\alpha,1,
\end{multline}
with the normalization condition
\begin{equation}
Y_\infty(\lambda)\underset{\lambda\to\infty}
=\left\{I+\sum_{k=1}^\infty\Phi_k(\infty)\left(\frac1{\lambda}\right)^k\right\}
\begin{pmatrix}
\lambda^{-\theta_\infty/2} & 0 \\
0 & \lambda^{\theta_\infty/2}
\end{pmatrix},
\end{equation}
where $\theta_\infty=-(s+q)$, $\det\Phi_0(\nu)=1$, and the matrices
$\mathcal{C}_\nu$, $\nu=0,1,\alpha$, are
\begin{equation}
\mathcal{C}_0=G_{C_0}^{-1}G_{C_\infty}^{-1},
\qquad \mathcal{C}_\alpha=\mathcal{C}_1=G_{C_\infty}^{-1}.
\end{equation}
\end{proposition}

\begin{remark}\label{Rem:triang}
It is clear that the monodromy matrices
\begin{equation}
M_\nu=e^{i\pi\theta_\nu}I,\qquad \nu=0, \alpha, 1,\infty,
\end{equation}
do not uniquely define $Y_\infty(\lambda)$. So, in this specific case we call the
monodromy data $\mathcal{C}_\nu$. Note that
because integers $\theta_0>0$, $\theta_1>0$, and $\theta_\alpha<0$,
the matrices $\mathcal{C}_\nu$ are defined by asymptotic expansions
\eqref{Yinfty} up to the left multiplication:
\begin{equation}\label{triang}
\mathcal{C}_0\to
\begin{pmatrix}
a_0& b_0 \\ 0 & 1/a_0
\end{pmatrix}\mathcal{C}_0,
\quad
\mathcal{C}_1\to
\begin{pmatrix}
a_1& b_1 \\ 0 & 1/a_1
\end{pmatrix}\mathcal{C}_1,
\quad
\mathcal{C}_\alpha\to\begin{pmatrix}
a_\alpha& 0 \\ b_\alpha & 1/a_\alpha
\end{pmatrix}\mathcal{C}_\alpha,
\end{equation}
where $a_\nu,b_\nu\in\mathbb{C}$, $\nu=0,1,\alpha$.
\end{remark}

Note that our singular matrix Riemann-Hilbert
problem on two contours can be reformulated in various ways
as different singular and regular Riemann-Hilbert problems.
In particular, one can formulate the Riemann-Hilbert problem
on an oriented simple closed contour $\gamma$ on $\mathbb{C}$,
which surrounds a domain containing the points
$\lambda=0$ and $\lambda=\alpha$ but excluding the point $\lambda=1$,
and with the jump matrix given by $G_\gamma=G_{C_0}G_{C_\infty}$. This can be
done in two ways: first, by performing analytic continuation of the solutions
$Y_\infty(\lambda)$ and $Y_0(\lambda)$ (the solution which is
analytic inside the disc surrounded by the contour $C_0$), and, second, just by
observing that in the Fredholm kernel $\hat K$ one can merge
the contours $C_0$ and $C_\infty$ into one contour $\gamma$
and then repeat the whole construction of this section literally.
The equivalence of these two procedures can be established
by using \eqref{triang}.

As a byproduct of our discussion above, one immediately obtains, with the help of
\eqref{matY}, the following formula
\begin{equation}\label{dlX}
\frac{\rmd}{\rmd \lambda} X(\lambda)=
A(\lambda)X(\lambda)-X(\lambda)\left(\frac{s}{2\lambda}S(0)
+\frac{r+q}{2(\lambda-1)}S(1)-\frac{r}{2(\lambda-\alpha)}S(\alpha)\right).
\end{equation}
This formula is crucial in derivation of relation \eqref{dadetK} below.
Note that \eqref{dlX} also be derived directly,
just by using the expressions
for the derivatives of functions $e_\pm(\lambda)$ and $f_\pm(\lambda)$
entering the definition of the matrix $X(\lambda)$.
Then expression \eqref{dlX} yields an alternative
prove of the system \eqref{eqlambda}, \eqref{eqalpha}, by noticing
that the second term
would not appear if instead of $X(\lambda)$ there would stand
the matrix $X(\lambda)$ times some
proper matrix factor, multiplied from the right.
Taking into account that the matrix $S(\lambda)$ given by \eqref{Sla}
admits factorization \eqref{Sfactor}, it is easy to see
that such a matrix factor can be chosen as
in the expression \eqref{matY} for the matrix $Y(\lambda)$.

\subsection{The Fredholm determinant as the $\sigma$-function}
\label{Sec:FDR}

To identify connection of the Fredholm determinant \eqref{EFPasFdet} with
P6, we consider the logarithmic
derivative of the determinant with respect to the parameter $\alpha$,
\begin{align}
\frac{\rmd}{\rmd\alpha}\log\Det \left(1-\hat K\right)
&=
-\Tr\left[\left(1-\hat K\right)^{-1}\frac{\rmd}{\rmd\alpha}\hat K\right]
\notag\\
&=-\Tr\left[\left(1+\hat R\right) \frac{\rmd}{\rmd\alpha}\hat K\right]
\notag\\
&=-\oint_{C_0}
\left[\left(1+\hat R\right)\frac{\rmd}{\rmd\alpha}K(\cdot,\mu)\right](\mu)\,
\rmd\mu.
\end{align}
To evaluate these integrals, we first note that the $\alpha$-derivative of the
kernel $K(\lambda,\mu)$, see \eqref{daK},
can be written in terms of the functions $g_\pm(\lambda;\alpha)$,
defined in \eqref{gpm}, as follows
\begin{multline}\label{daKgg}
\frac{\rmd}{\rmd\alpha}K(\lambda,\mu)
=-\frac{r}{4\pi\rmi}
\big[g_{+}(\lambda;\alpha)g_{-}(\mu;\alpha)
\\
+g_{-}(\lambda;\alpha)g_{+}(\mu;\alpha)
-2E(\alpha)g_{-}(\lambda;\alpha)g_{-}(\mu;\alpha)\big].
\end{multline}
Introducing a notation (see also \eqref{gpm})
\begin{equation}
\vec g_{*}(\lambda;\nu)=
\begin{pmatrix}
g_{-}(\lambda;\nu) \\ - g_{+}(\lambda;\nu)
\end{pmatrix},
\end{equation}
we note that the linear combination of the functions
standing in the brackets in \eqref{daKgg} can also be
written as
\begin{align}
\frac{\rmd}{\rmd\alpha}K(\lambda,\mu)
&=-\frac{r}{4\pi\rmi}\vec g_{*}^T(\mu;\alpha)
S(\alpha)\vec g(\lambda;\alpha)
\notag\\
&=-\frac{r}{4\pi\rmi}\tr\left[\vec g(\lambda;\alpha)
\vec g_{*}^T(\mu;\alpha)
S(\alpha)\right],
\end{align}
where the matrix $S(\lambda)$ is given in \eqref{Sla}.
Then, using the result of action of the operator $1+\hat R$ on the
functions $g_\pm(\lambda;\nu)$ in its vector form  \eqref{gvecf}, we have
\begin{equation}
\left[\left(1+\hat R\right)\frac{\rmd}{\rmd\alpha}K(\cdot,\mu)\right](\lambda)
=-\frac{r}{4\pi\rmi(\mu-\alpha)}
\tr\left[X^{-1}(\alpha)\vec f(\lambda;\alpha)\vec g_{*}^T(\mu;\alpha)S(\alpha)\right].
\end{equation}
Setting here $\lambda=\mu$ and integrating over $\mu$ gives
\begin{align}
\oint_{C_0}
\frac{1}{\mu-\alpha}
\vec f(\mu;\alpha)\vec g_{*}^T(\mu;\alpha)\frac{\rmd \mu}{2\pi\rmi}
&=
\oint_{C_0}
\frac{1}{(\mu-\alpha)^2}
\begin{pmatrix}
f_{+}(\mu)e_{-}(\mu) & - f_{+}(\mu)e_{+}(\mu) \\
f_{-}(\mu)e_{-}(\mu) & - f_{-}(\mu)e_{+}(\mu)
\end{pmatrix}
\frac{\rmd \mu}{2\pi\rmi}
\notag\\ &
=\frac{\rmd}{\rmd \lambda}
\begin{pmatrix}
H_{-+}(\lambda) & - H_{++}(\lambda) \\
H_{--}(\lambda) & - H_{+-}(\lambda)
\end{pmatrix}\bigg|_{\lambda=\alpha}
\notag\\ &
=\frac{\rmd}{\rmd \lambda} X(\lambda)\bigg|_{\lambda=\alpha}.
\end{align}
Hence, we obtain the following expression for  logarithmic derivative of the
Fredholm determinant:
\begin{equation}\label{tr2}
\frac{\rmd}{\rmd\alpha}\log\Det \left(1-\hat K\right)=
-\frac{r}{2}
\tr\left[ S(\alpha) X^{-1}(\alpha) \frac{\rmd}{\rmd \lambda} X(\lambda)
\right]\bigg|_{\lambda=\alpha}.
\end{equation}

Using the $\lambda$-derivative equation for the
formal matrix solution, one can express the matrix $X'(\lambda)$
in terms of matrices $X(\lambda)$ and $A(\lambda)$, and evaluate
the limit $\lambda\to\alpha$ for the trace. In doing so we
use the fact that even though the matrix $X'(\lambda)$
has a pole at $\lambda=\alpha$, the trace in \eqref{tr2}
does not contain the $(\lambda-\alpha)^{-1}$ term. To see this, it is
useful to rewrite \eqref{dlX} in somewhat more explicit form as
\begin{equation}
\frac{\rmd}{\rmd \lambda}X(\lambda)=\sum_{\nu=0,1,\alpha}\frac{\theta_\nu}{2(\lambda-\nu)}
\left(X(\nu)S(\nu)X^{-1}(\nu)X(\lambda)-X(\lambda)S(\nu)\right).
\end{equation}
Keeping $\lambda$ arbitrary, one can see that the $\nu=\alpha$ term does
not contribute to the trace, since
\begin{multline}
\tr\left[S(\alpha)X^{-1}(\alpha)
\left(X(\alpha)S(\alpha)X^{-1}(\alpha)X(\lambda)-X(\lambda)S(\alpha)\right)\right]
\\
=\tr\left[S^2(\alpha)X^{-1}(\alpha)X(\lambda)-S(\alpha)X^{-1}(\alpha)X(\lambda)S(\alpha)\right]=0.
\end{multline}
Setting $\lambda=\alpha$ for the remaining two terms and writing
$\theta_\alpha/2$ for the pre-factor in \eqref{tr2}, we have
\begin{align}
&
\frac{\rmd}{\rmd\alpha}\log\Det \left(1-\hat K\right)
\notag\\ & \qquad
=\frac{\theta_\alpha}{2}
\tr\bigg[S(\alpha)X^{-1}(\alpha)
\bigg(\frac{A_0 X(\alpha)}{\alpha}-\frac{\theta_0 X(\alpha)S(0)}{2\alpha}
+\frac{A_1 X(\alpha)}{\alpha-1}-\frac{\theta_1 X(\alpha)S(1)}{2(\alpha-1)}\bigg)
\bigg]
\notag\\ & \qquad
=\tr\left[A_\alpha\left(\frac{A_0}{\alpha}+\frac{A_1}{\alpha-1}\right)\right]
-\frac{\theta_\alpha}{4}\tr\left[S(\alpha)\left(\frac{\theta_0 S(0)}{\alpha}+
\frac{\theta_1 S(1)}{\alpha-1}\right)\right].
\end{align}
Taking into account that the matrix $S(\lambda)$ is upper triangular, see \eqref{Sla},
we obtain
\begin{equation}\label{dadetK}
\frac{\rmd}{\rmd\alpha}\log\Det \left(1-\hat K\right)
=\tr\left[A_\alpha\left(\frac{A_0}{\alpha}+\frac{A_1}{\alpha-1}\right)\right]
-\frac{\theta_\alpha}{2}\left(\frac{\theta_0}{\alpha}+\frac{\theta_1}{\alpha-1}\right).
\end{equation}
Using \eqref{nuks}, thus we have for the trace
\begin{equation}
\tr\left[A_\alpha\left(\frac{A_0}{\alpha}
+\frac{A_1}{\alpha-1}\right)\right]
=\frac{\rmd}{\rmd\alpha}\log\Det \left(1-\hat K\right)
+\frac{(\nu_1+\nu_2)(-2\nu_3\alpha+\nu_3+\nu_4)}{2\alpha(\alpha-1)},
\end{equation}
where $\nu_k$'s are defined in \eqref{nuks}.

Substituting the expression for the trace above into \eqref{sigmatr}, we arrive at
the following expression for the $\sigma$-function
\begin{multline}
\sigma=\alpha(\alpha-1)\frac{\rmd}{\rmd\alpha}\log\Det\left(1-\hat K\right)
+(\nu_1\nu_2-\nu_1\nu_3-\nu_2\nu_3)\alpha
\\
-\frac{\nu_1\nu_2+\nu_3\nu_4-(\nu_1+\nu_2)(\nu_3+\nu_4)}{2}.
\end{multline}
Taking into account that from our previous discussion it follows that
\begin{equation}
\theta_0=s,\quad
\theta_1=r+q,\quad
\theta_\alpha=-r,\quad
\theta_\infty=-(\theta_0+\theta_1+\theta_\alpha)=-(s+q),
\end{equation}
the $\nu_k$'s are therefore identified to be
\begin{equation}\label{ournuks}
\nu_1=\nu_3=-\frac{r+q+s}{2},\quad
\nu_2=-\frac{r-q-s}{2},\quad
\nu_4=\frac{r+q-s}{2}.
\end{equation}
Using these values, for the $\sigma$-function we obtain the following
expression
\begin{equation}
\sigma=\alpha(\alpha-1)
\frac{\rmd}{\rmd\alpha}\log\Det \left(1-\hat K\right)
-\frac{(r+q+s)^2}{4}\alpha+\frac{(r+q+s)q+2rs}{4}.
\end{equation}
It satisfies \eqref{eqsigma} with the $\nu_k$'s
as specified in \eqref{ournuks}. Since the EFP is represented in
as the Fredholm  determinant, $F_{r,s,q}=\Det(1-\hat K)$, we conclude
that we arrive at the statement of Theorem~\ref{EFPasP6},
which is thus proved.

\section{Asymptotic expansions of the EFP in the thermodynamic limit}

Under the thermodynamic limit of a physical quantity given on a lattice
one usually understands behavior of this quantity as the size of the lattice becomes large.
Specifically, in the case of the EFP besides the size of the lattice, $N=r+s+q$,
there are geometric parameters,  $s$ and $q$, which describe the frozen
rectangle at the top left corner of the lattice.
Therefore it is interesting to consider the thermodynamic limit $N\to\infty$
preserving the geometry of the problem, i.e., keeping the ratios
$s/N$ and $q/N$ fixed.  Here we consider solution
of this problem in the important particular case where $q=0$;
it will be convenient for us to fix the ratio $s/r$ and perform the limit
$s\to\infty$.
We show that this problem can be solved in a systematic way
using the $\sigma$-form of P6 and the asymptotic behavior of $F_{r,s,q}$
at the singular points of this equation.

We begin with discussion of physical interpretation
of leading terms of asymptotic expansions in two different regimes.
The details of calculations of the asymptotic behavior of $F_{r,s,q}$ in each regime follow next.
In Appendices we give an alternative derivation of the asymptotics in one of the regimes.

\subsection{Two asymptotic regimes}\label{Sec:tworegimes}

Recall that the EFP is defined as a probability
of observing the certain arrow configuration in the six-vertex model with DWBC.
It can be seen as a ratio of two partition functions:
$F_{r,s,q}=Z_{r,s,q}/Z_N$, where  $Z_{r,s,q}$ and $Z_N$ are
the partition functions of the model
on the lattice with the frozen corner, and on the original, unmodified lattice,
respectively. Since both partition functions in the thermodynamic
limit behave as $\exp(-V f)$ where $V$ is the volume (the number of lattice sites)
and $f$ is the free-energy per site, we conclude that in the limit, where
$r,s,q\to \infty$, we should have
\begin{equation}\label{s2decay}
F_{r,s,q}= \exp\left\{-\varphi s^2 +o(s^2)\right\}.
\end{equation}
The quantity $\varphi$ is the function of the ratios $r/s$ and
$q/s$, and the parameter $\alpha$; moreover, $\varphi\geq0$
since $Z_{r,s,q}\leq Z_N$.
This function has the meaning of
the change of the free-energy per site
due to the freezing of a macroscopically large rectangle at the corner of the lattice.
It is clear that if $s>r$, i.e., the bottom-right corner of the frozen rectangle
lies below the counter-diagonal of the square, then, as follows from the ice-rule,
$F_{r,s,q}=0$ and hence $\varphi=\infty$.
At $s=r$, it can be shown that for generic
weights $F_{r,s,q}=(w_1w_2)^{s(s+q)}Z_rZ_{r+q}/Z_N$; in our case of the free-fermion
weights \eqref{FF} $Z_n=1$ for all $n\in\mathbb{N}$, that follows from the results
of papers \cite{I-87,EKLP-92}, hence $F_{r,s,q}=(1-\alpha)^{s(s-1)}$.
At $s=0$, obviously, $F_{r,s,q}=1$, hence $\varphi=0$. Thus the problem
consists in studying the EFP in the thermodynamic limit for $0<s<r$.

It is clear that $\varphi$ is monotonically non-decreasing function for $s\in(0,r)$.
Thus, there may exist a region of values of
$s$, adjacent to zero, where $\varphi\equiv 0$. In this case the quantity
$o(s^2)$ in \eqref{s2decay} decays much faster than $s^2$, namely,
it is exponentially small as $s\to\infty$. In this region
the EFP is close to $1$, and we have the following behavior
\begin{equation}\label{s1decay}
F_{r,s,q}=1-\exp\left\{-\chi s+o(s)\right\},\qquad
\chi\geq0.
\end{equation}
Note that this kind of behavior can be expected if we recall that the EFP can be treated
as an $s$-point quantum correlation function \cite{CP-07a} and that
ferroelectric correlations in the six-vertex model are exponentially decaying.

It is to be also mentioned that our EFP has the meaning of a
cumulative distribution function in the context of the random growth model,
considered by Johansson in \cite{J-00}, who proved the existence of the
two asymptotic regimes \eqref{s2decay} and \eqref{s1decay}.
The functions $\varphi$ and $\chi$ are
called there as the lower and upper tail deviation functions, respectively.

The existence of the two asymptotic regimes \eqref{s2decay} and \eqref{s1decay}
is a manifestation of the spacial separation of phases of the six-vertex model
with DWBC in the thermodynamic limit, the so-called arctic ellipse phenomenon.
Exactly the same phenomenon is often mentioned in the literature in relation to the
dimer problem on an Aztec diamond graph (or domino tilings of an Aztec diamond)
\cite{EKLP-92,JPS-98}.

The phenomenon consists in a peculiar form of statistically dominating configurations
of the six-vertex model with DWBC on the $N\times N$ lattice as $N\to\infty$.
The arctic ellipse arises when $N\times N$ lattice is scaled
to a square of unit side length, and it divides this square on two types of domains:
the interior of the ellipse and the four corner domains, see Fig.~\ref{fig-aell}.
The local states (arrows) are disordered in the
interior of the ellipse, while they are ordered in the each of the four
corner domains due to strong influence of the boundary conditions.
Recall that the EFP describes the probability that all arrows are
ordered (frozen) in the rectangle at the top left corner of the lattice.
There are two principal cases, whether the rectangle intersects the arctic
ellipse or not. When the rectangle intersects the arctic ellipse one
may expect that the EFP behaves as given by \eqref{s2decay}.
Since in this case the rectangle contains a portion of the domain of
the disorder, we will say that then the EFP is in the \emph{disordered} regime.
If the rectangle does not intersects the arctic ellipse
then one may expect that the EFP behaves as given by \eqref{s1decay}.
Since in this case the rectangle encloses only a region of order
we will say that the EFP is in the \emph{ordered} regime.
Indeed, when the rectangle does not intersects the arctic ellipse
it does not essentially influence on the number
of the states of the system; thus the EFP is close to $1$ for large $N$, and, in fact,
equals to $1$ in the thermodynamic limit.
Contrary, when it intersects the arctic ellipse it starts to order the
disordered domain and reduces significantly the number of states, so that the EFP
should be small for large $N$, and, in fact, vanishes in the limit.
In Fig.~\ref{fig-aell} these two situations are shown in the case $q=0$, i.e., when
the rectangle is just square.

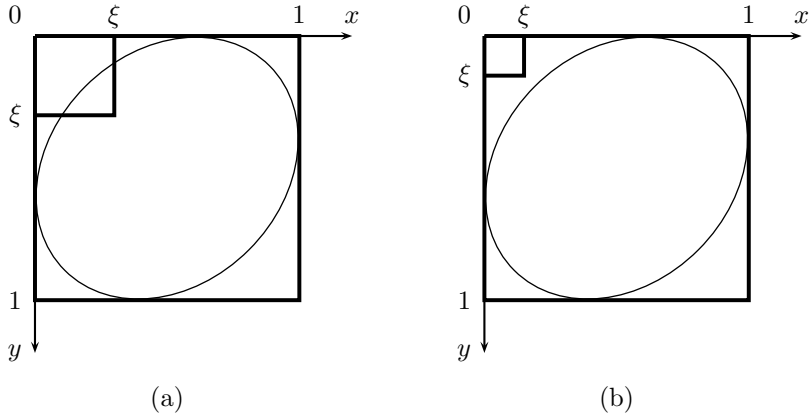
\begin{figure}
\centering

\psset{unit=10pt}
\begin{pspicture}(-1,-4)(29,11)
\psline(10,10)(10,0)(0,0)
\psline{->}(0,10)(12,10)
\psline{->}(0,10)(0,-2)
\rput[r](-.5,-2){$y$}
\rput[B](12,10.5){$x$}
\rput{45}(5,5){\psellipse[linewidth=.05](0,0)(5.5,4.4)}
\rput[rB](-.5,10.5){$0$}
\rput[B](10,10.5){$1$}
\rput[r](-.5,0){$1$}
\rput[B](3,10.5){$\xi$}
\rput[r](-.5,7){$\xi$}
\psline[linewidth=.15](0,7)(3,7)(3,10)
\psline[linewidth=.15](0,10)(0,0)(10,0)(10,10)(0,10)
\rput[B](5,-4){(a)}

\rput(17,0){
\psline(10,10)(10,0)(0,0)
\psline{->}(0,10)(12,10)
\psline{->}(0,10)(0,-2)
\rput[r](-.5,-2){$y$}
\rput[B](12,10.5){$x$}
\rput{45}(5,5){\psellipse[linewidth=.05](0,0)(5.5,4.4)}
\rput[rB](-.5,10.5){$0$}
\rput[B](10,10.5){$1$}
\rput[r](-.5,0){$1$}
\rput[B](1.5,10.5){$\xi$}
\rput[r](-.5,8.5){$\xi$}
\psline[linewidth=.15](0,8.5)(1.5,8.5)(1.5,10)
\psline[linewidth=.15](0,10)(0,0)(10,0)(10,10)(0,10)
\rput[B](5,-4){(b)}
}
\end{pspicture}

\caption{The $N\times N$ lattice scaled to the unit square in the large $N$ limit,
the frozen square of side length $\xi=\frac{s}{N}=\frac{v}{1+v}$,
and the arctic ellipse $\frac{(1-x-y)^2}{1-\alpha}+\frac{(x-y)^2}{\alpha}=1$
separating the phases: (a) For $v\in(\vc,1)$ the frozen square overlaps
the interior of the ellipse, in this case the EFP is in the disordered regime;
(b) For $v\in(0,\vc)$ the frozen square lies completely outside of the ellipse,
in this case the EFP is in the ordered regime.}
\label{fig-aell}
\end{figure}

In calculations we limit ourselves here to the case $q=0$.
In this case the analytic results appear considerably simpler
without altering the main ingredients of the method.
We take the two remaining integer parameters, $s$ and $r$, $0< s< r$
to be large in the limit, with their ratio being a non-vanishing parameter
\begin{equation}
v\equiv\frac{s}{r},\qquad v\in (0,1).
\end{equation}
The condition of intersection of the square with the arctic ellipse
reads
\begin{equation}\label{vc}
v=\vc,\qquad \vc\equiv\frac{1-\sqrt\alpha}{1+\sqrt\alpha}.
\end{equation}
For $v \in (\vc,1)$ the EFP is in the disordered regime, and for
$v(0,\vc)$ it is in the ordered one. Sometimes instead of natural
variables $v$ and $\vc$, it is useful to use the variables $\alpha$
and $\beta$, in which the condition of intersection of the ellipse
reads
\begin{equation}\label{beta}
\alpha=\beta,\qquad \beta\equiv\left(\frac{1-v}{1+v}\right)^2.
\end{equation}
For $\alpha\in (\beta,1)$ the EFP is in the disordered regime, and for
$\alpha\in(0,\beta)$ it is in the ordered one.

\subsection{Asymptotic expansion in the disordered regime}
\label{Sec:Disorder}

We begin with the study of the EFP in the disordered regime,
which corresponds to $v\in(\vc,1)$ or $\alpha\in(\beta,1)$,
by considering behavior of $F_{r,s,q}$ at the critical point
$\alpha=1$ of P6.

It was noticed in \cite{CP-13} that
$F_{r,s,0}\sim C_{r,s} (1-\alpha)^{s^2}$ where
the constant $C_{r,s}$ is essentially given by the value
of the Hankel determinant in \eqref{newEFP} at $\alpha=1$.
However, this asymptotics does not uniquely
characterize the corresponding solution of the $\sigma$-form of P6,
and therefore we need a more
elaborated result.

\begin{proposition}\label{prop:al-to-1}
As $\alpha\to 1$, the EFP behaves as follows
\begin{multline}\label{Frs}
F_{r,s,0}=C_{r,s} (1-\alpha)^{s^2}
\bigg(1-\frac{s(r-s)}{2}(1-\alpha)
\\
+
\frac{s(r-s)\left(2s^3r-2s^4-3s^2+1\right)}{4(4s^2-1)}(1-\alpha)^2+O
\left((1-\alpha)^3\right)\bigg),
\end{multline}
where the constant $C_{r,s}$ has the explicit form
\begin{equation}\label{Crs}
C_{r,s}=
\prod_{j=0}^{s-1}
\frac{(j+r)!(j!)^2}{(r-j-1)! (2j)! (2j+1)!}.
\end{equation}
\end{proposition}
\begin{proof}
At $\alpha=1$ the Hankel determinant in \eqref{newEFP}
appears to be related in the canonical way with an ensemble of Hahn polynomials
\cite{CP-13,CP-15}. At $q=0$ we deal with the special case of the Hahn polynomials
$Q_n(x;0,0,r-1)$; for the notation and basic properties see, e.g., \cite{KLS-10}.
We will use the normalized version of these polynomials
with the coefficient of the highest-order term equal to $1$,
\begin{equation}
p_n(x)=\frac{(1)_n (-r+1)_n}{(n+1)_n}Q_n(x;0,0,r-1)=
\frac{(1)_n (-r+1)_n}{(n+1)_n}
\Fthreetwo{-n}{n+1}{-x}{1}{-r+1}{1},
\end{equation}
where $(a)_n=a(a+1)\cdots(a+n-1)$ is the Pochhammer symbol.
These polynomials satisfy the orthogonality condition
\begin{equation}\label{orthcond}
\sum_{x=0}^{r-1}p_n(x)p_{n'}(x)=h_n\delta_{nn'},\qquad
h_n=\frac{(n!)^4 (r+n)!}{(2n)!(2n+1)!(r-n-1)!}
\end{equation}
and the recurrence relation
\begin{equation}\label{recrel}
xp_n(x)=p_{n+1}(x)+\frac{r-1}{2}p_n(x)+\frac{h_n}{h_{n-1}}p_{n-1}(x),
\qquad \frac{h_n}{h_{n-1}}=\frac{n^2(r^2-n^2)}{4(4n^2-1)}.
\end{equation}
Let us denote the Hankel matrix in the representation \eqref{newEFP}
by $H(\alpha)$; we also write for brevity $H\equiv H(1)$.
At $q=0$ (for the case of $q\ne0$, see \cite{CP-15})
the entries of the matrix $H$ appear to be
the moments of the orthogonality measure in \eqref{orthcond} and therefore
\begin{equation}
\det H = \det_{1\leq j,k\leq s}\left[\sum_{m=0}^{r-1}m^{j+k-2}\right]=
\prod_{j=0}^{s-1} h_j= \prod_{j=0}^{s-1}\frac{(j!)^4 (r+j)!}{(2j)!(2j+1)!(r-j-1)!}.
\end{equation}
Taking into account the value of the prefactor in \eqref{newEFP}, we obtain
the expression \eqref{Crs} for the constant $C_{r,s}$.

Let us now focus on the terms of the $(1-\alpha)$-expansion in \eqref{Frs}. We denote
$H'\equiv \frac{\rmd}{\rmd \alpha} H(\alpha)|_{\alpha=1}$
and $H''\equiv \frac{\rmd^2}{\rmd \alpha^2} H(\alpha)|_{\alpha=1}$.
Using the fact that $\det H(\alpha)=\exp(\tr\log H(\alpha))$ and differentiating,
one can show that the coefficients of the Taylor expansion
\begin{equation}
\det H(\alpha) = \det H
\cdot\left(1+c_1(1-\alpha)+\frac{c_2}{2}(1-\alpha)^2+O\left((1-\alpha)^3\right)\right),
\end{equation}
can be written as
\begin{equation}
c_1=-\tr(H'H^{-1}),\quad
c_2=(\tr(H'H^{-1}))^2 - \tr(H'H^{-1})^2+\tr(H''H^{-1}).
\end{equation}
Now we use the fact that
\begin{equation}
H^{-1}_{jk}=\frac{1}{(j-1)!}
\frac{\partial^{j-1}}{\partial x^{j-1} }
\frac{1}{(k-1)!}
\frac{\partial^{k-1}}{\partial y^{k-1} }
\sum_{n=0}^{s-1} \frac{p_n(x)p_n(y)}{h_n}\bigg|_{x=y=0}.
\end{equation}
Taking into account
that $H_{jk}'=\sum_{m=0}^{r-1}m^{j+k-1}$ and $H_{jk}''=\sum_{m=0}^{r-1}m^{j+k-1}(m-1)$,
one can reduce evaluation of the traces in the expressions above to
the standard sums involving the related polynomials, namely
\begin{equation}
\sum_{m=0}^{r-1}\frac{m p_n^2(m)}{h_n}=\frac{r-1}{2},\qquad
\sum_{m=0}^{r-1}\frac{m^2 p_n^2(m)}{h_n}=\frac{h_{n+1}}{h_n}
+\left(\frac{r-1}{2}\right)^2+
\frac{h_{n}}{h_{n-1}},
\end{equation}
which can be easily derived from the orthogonality condition
\eqref{orthcond} and the recurrence relation \eqref{recrel}.
In this way, for the coefficient $c_1$ we get
\begin{equation}
c_1=-\sum_{n=0}^{s-1}\sum_{m=0}^{r-1}\frac{m p_n^2(m)}{h_n}=-\frac{s(r-1)}{2}.
\end{equation}
For the coefficient $c_2$, after the straightforward but more involved calculations,
we obtain
\begin{equation}
c_2=-\frac{s(r-1)}{2}+\frac{s^2(1-2r+s^2)}{4}+\frac{s^4(r^2-s^2)}{4s^2-1}.
\end{equation}
Finally, taking into account that
\begin{equation}
\frac{1}{\alpha^{s(s-1)/2}}=1+\frac{s(s-1)}{2}(1-\alpha)
+\frac{s(s-1)(s^2-s+2)}{8}(1-\alpha)^2+O\left((1-\alpha)^3\right),
\end{equation}
we arrive at \eqref{Frs}.
\end{proof}

As a direct consequence of \eqref{Frs}, we have the following.
\begin{corollary}
As $\alpha\to1$,
\begin{multline}\label{logFrs0}
\log F_{r,s,0}=\log C_{r,s}+ s^2\log(1-\alpha)
-\frac{s(r-s)}{2}(1-\alpha)
\\
-\frac{s(r-s)\left(7s^2-s r -2\right)}{8(4s^2-1)}(1-\alpha)^2+O
\left((1-\alpha)^3\right).
\end{multline}
\end{corollary}

To construct the asymptotics
of $\log F_{r,s,0}$ for large $s$ and $r=s/v$, we need the corresponding
behavior of the coefficients of \eqref{logFrs0}. A nontrivial calculation
concerns only
$\log C_{r,s}$. To do that we rewrite \eqref{Crs} with the help of the
Barnes $G$-function \cite{B-1900} as follows:
\begin{equation}\label{CGG}
C_{r,s}=
\frac{\pi^{s+1/2}\, G^2(1/2)G(r+s+1)G(r-s+1)}{2^{s(2s-1)}\,G^2(r+1)G(s+1/2) G(s+3/2)}.
\end{equation}
The Barnes G-function satisfy the relations
\begin{equation}
G(z+1)=\Gamma(z) G(z), \qquad
G(1)=G(2)=G(3).
\end{equation}
As established by Barnes:
\begin{multline}\label{largeBarnes}
\log G(z+1)=\frac{z^2}{2}\log z -\frac{3}{4}z^2
+\frac{\log 2\pi}{2} z-\frac{1}{12}\log z
+\zeta'(-1)
\\
+\sum_{k=1}^{n}\frac{B_{2k+2}}{4k(k+1)}z^{-2k}+O(z^{-2n-2}),
\end{multline}
where $\zeta'(-1)=-0.165142...$ is the derivative of the Riemann function $\zeta(z)$
at $z=-1$, and $B_{2k}$'s are the Bernoulli numbers:
\begin{equation}\label{Bernoulli}
B_0=1,\quad
B_2=\frac{1}{6},\quad
B_4=-\frac{1}{30},\quad
B_6=\frac{1}{42}, \quad
B_8=-\frac{1}{30}, \quad
B_{10}=\frac{5}{66}, \quad \ldots
\end{equation}

Therefore, taking the logarithm of \eqref{CGG} and using \eqref{largeBarnes},
we arrive at the following result.
\begin{proposition}
As $s,r\to \infty$, with $s/r=v$ fixed, $v\in[0,1]$, we have
\begin{multline}\label{logC}
\log C_{r,s}=-
s^2\left(\log4v-\frac{(1+v)^2}{2v^2}\log(1+v)-\frac{(1-v)^2}{2v^2}\log(1-v)\right)
\\
-\frac{1}{12}\log s -\frac{1}{12}\log\frac{(1-v^2)}{2} +\zeta'(-1)
+\sum_{k=1}^{n} \frac{\mathcal{B}_{2k}}{s^2}+O(s^{-2n-2}),
\end{multline}
where the coefficients $\mathcal{B}_{2k}=\mathcal{B}_{2k}(v)$ are
\begin{multline}
\mathcal{B}_{2k}=
\frac{B_{2k+2}}{4k(k+1)}\left(\frac{v^{2k}}{(1+v)^{2k}}
+\frac{v^{2k}}{(1-v)^{2k}}-2v^{2k}\right)
\\
-\sum_{m=0}^{k-1}\frac{B_{2(k-m)+2}}{2^{2m+1}(k-m)(k-m+1)}\binom{2k-1}{2m}
-\frac{4k^2+6k-1}{2^{2k+3}3k(k+1)(2k+1)},
\end{multline}
where $\binom{2k-1}{2m}$ are Binomial coefficients.
For example,
\begin{equation}
\mathcal{B}_2=-\frac{1}{8}\left(\frac{1}{8}-\frac{1+v^2}{15}+
\frac{v^2(1+v^2)}{15(1-v^2)^2}\right),\quad
\mathcal{B}_4=-\frac{v^6(v^6-4v^4+5v^2-10)}{504(1-v^2)^4}+\frac{31}{16128}.
\end{equation}
\end{proposition}

Note that, besides single logarithmic term, the expansion
for $\log C_{r,s}$ involves only even powers of $1/s$.
The same property appears to be valid for the
explicit terms in \eqref{logFrs0}.

To derive expansions of the EFP in the thermodynamic limit, we start with
analyzing the $\sigma$-form of P6 in the large $s$ limit. We note that
for $q=0$ the $\sigma$-function in \eqref{eq:sigma-frsq} reads
\begin{equation}\label{qzero}
\sigma=\left(-\frac{(1+v)^2}{4v^2}\alpha+\frac{1}{2v}\right)s^2+
\alpha(\alpha-1)\frac{\rmd}{\rmd\alpha} \log F_{r,s,0}.
\end{equation}
Absence of the odd powers of $s$ in the $\sigma$-form of P6
\eqref{sigmaform} together with the expansions \eqref{logFrs0} and \eqref{logC}
suggests that the $\sigma$-function may be searched
in the form of the following asymptotic
ansatz in the decaying powers of $s^2$:
\begin{equation}\label{sigmas}
\sigma= \sigma_2 s^2 +
\sigma_0+\sigma_{-2}s^{-2}+\cdots,\qquad s\to\infty.
\end{equation}
Note that if we succeed in the construction
of the expansion \eqref{sigmas} with the coefficients
$\sigma_{2k}=\sigma_{2k}(\alpha)$, which are analytical functions
of $\alpha$, then the Wazow theorem, see Th.~36.1 in Chap.~IX of  \cite{W-87},
implies that there exists a genuine solution with asymptotics \eqref{sigmas}.
To justify that this solution actually coincides with
the solution given by \eqref{qzero} one can verify
that both solutions have the same
asymptotics as $\alpha\to 1$. Here we present the general scheme of derivation 
of asymptotics; more details are given in Remark~\ref{RemMethod} below. 

The expansion \eqref{sigmas} can be constructed in a standard
way by substituting it in \eqref{sigmaform}. On this way we
first obtain the leading term, by requiring that it reproduces the
asymptotic condition at $\alpha\to1$, and next prove
that the further terms can be obtained recursively.

It happens that the leading term $\sigma_2$ can be found as
a solution of the first-order ODE which is obtained by
noticing that the first term in
the left-hand side of \eqref{sigmaform}, which contains
second-order derivatives, is sub-leading as $s$ large.
Besides the overall factor $\sigma_2'$,
the resulting first-order ODE splits on two
equations
\begin{equation}\label{reducedODE}
\frac{1+v^2}{4v^2}+\sigma_2+(1-\alpha)\sigma_2'=0,\qquad
\frac{(1-v^2)^2}{16v^4} + \sigma_2'
\left(\frac{1+v^2}{4v^2}-\sigma_2+\alpha \sigma_2'\right)=0.
\end{equation}
These equations have the following general solutions
\begin{equation}\label{gensol}
(\sigma_2)_\mathrm{I}=C_\mathrm{I}\,(\alpha-1) -\frac{1+v^2}{4v^2},
\qquad
(\sigma_2)_\mathrm{II}
=C_\mathrm{II}\,\alpha +\frac{1+v^2}{4v^2}+\frac{(1-v^2)^2}{16
C_\mathrm{II}\, v^4},
\end{equation}
respectively.
Here $C_\mathrm{I}$ and $C_\mathrm{II}$ are integration constants,
which may depend on $v$.
Moreover, the second equation has two partial solutions
\begin{equation}\label{partsol}
(\sigma_{2})_\pm=\frac{1+v^2}{4v^2}\pm\frac{1-v^2}{2v^2}\sqrt\alpha.
\end{equation}

Consider now the function $\varphi$, defined by \eqref{s2decay}.
Using \eqref{qzero} and \eqref{sigmas} we find that
\begin{equation}\label{varphiint}
\varphi=\int \left(\sigma_2+\frac{(1+v)^2}{4v^2}\alpha
-\frac{1}{2v}\right)\frac{\rmd \alpha}{\alpha(1-\alpha)}+C,
\end{equation}
where $C=C(v)$ is an integration constant. The function $\varphi$
should obey the following asymptotic condition as $\alpha\to 1$:
\begin{multline}\label{phitoone}
\varphi=-\log(1-\alpha)
+\log4v-\frac{(1+v)^2}{2v^2}\log(1+v)-\frac{(1-v)^2}{2v^2}\log(1-v)
\\
+\frac{1-v}{2v}(1-\alpha)
-\frac{(1-v)(1-7v)}{32v^2}(1-\alpha)^2+ O\left((1-\alpha)^3\right).
\end{multline}

The function $(\varphi)_\mathrm{I}$ which corresponds to solution
$(\sigma_2)_\mathrm{I}$ reads
\begin{equation}
(\varphi)_\mathrm{I}
=-\left(C_\mathrm{I}+\frac{(1+v)^2}{4v^2}\right)\log\alpha +C.
\end{equation}
This function does not have asymptotics \eqref{phitoone},
because of the absence the leading term
$-\log(1-\alpha)$.

In the case of the function
$(\varphi)_\mathrm{II}$, which corresponds to
$(\sigma_2)_\mathrm{II}$, the leading term
$-\log(1-\alpha)$ of the expansion \eqref{phitoone}
dictates the choice
$C_\mathrm{II}=-\frac{1-v^2}{4v^2}$, then we get
\begin{equation}
(\varphi)_\mathrm{II}
=-\log(1-\alpha)-\frac{1-v}{2v}\log\alpha+C.
\end{equation}
As $\alpha\to1$,
\begin{equation}
(\varphi)_\mathrm{II}
=
-\log(1-\alpha)
+C+\frac{1-v}{2v}(1-\alpha)
+\frac{1-v}{4v}(1-\alpha)^2+ O\left((1-\alpha)^3\right).
\end{equation}
Therefore, because of the coefficient of the term $(1-\alpha)^2$,
the asymptotics \eqref{phitoone} cannot be satisfied
for any choice of $C$.

Now consider the functions $(\varphi)_{\pm}$ corresponding
to partial solutions \eqref{partsol},
\begin{equation}
(\varphi)_{\pm}
=-\log\left(1\pm\sqrt\alpha\right)
-\frac{1}{v^2}\log\left(1\mp\sqrt\alpha\right)
+\frac{(1-v)^2}{4v^2}\log \alpha+C.
\end{equation}
Asymptotics $\alpha\to1$ of the function $(\varphi)_{+}$
has the wrong coefficient of the leading term, namely,
$-1/v^2$, compare with \eqref{phitoone}.
For the function $(\varphi)_{-}$, as $\alpha\to1$,
we get
\begin{multline}
(\varphi)_{-}=
-\log(1-\alpha)+ C-\frac{1-v^2}{v^2}\log 2
\\
+\frac{1-v}{2v}(1-\alpha)
-\frac{(1-v)(1-7v)}{32v^2}(1-\alpha)^2+ O\left((1-\alpha)^3\right).
\end{multline}
Therefore, comparing the last asymptotics with \eqref{phitoone},
we observe
that $\varphi$ is given by the function
$(\varphi)_{-}$ with the proper choice of the integration constant $C$,
\begin{multline}\label{phires}
\varphi=-\log\left(\frac{1-\sqrt\alpha}{2}\right)
-\frac{1}{v^2}\log\left(\frac{1+\sqrt\alpha}{2}\right)
+\frac{(1-v)^2}{4v^2}\log \alpha
\\
+\log v-\frac{(1+v)^2}{2v^2}\log(1+v)-\frac{(1-v)^2}{2v^2}\log(1-v).
\end{multline}
Note that $\varphi$ is positive for $\alpha\in (\beta,1)$,
because it can be presented in the form
\begin{equation}\label{phi-int}
\varphi=\int_\beta^\alpha
\left(\frac{1-v}{2v}-\frac{1+v}{2v}\sqrt{\tilde\alpha}\right)^2
\frac{\rmd\tilde\alpha}{\tilde\alpha(1-\tilde\alpha)},
\end{equation}
where $\beta$ is given by \eqref{beta}.

Formula \eqref{phires}, when written in terms of the parameter $u$,
is exactly the expression \eqref{varphi}
and coincides with the leading term of the asymptotic expansion of the
EFP in the disordered regime obtained in \cite{CP-13}.

Thus we have just constructed the term $\sigma_2$ of the expansion
\eqref{sigmas}; the further terms $\sigma_{-2k}$,
$k=0,1,\ldots$, satisfy the recurrence
relation of the following form
\begin{equation}\label{recurr}
\left(v^2 (1+\sqrt{\alpha})^2-(1-\sqrt{\alpha})^2\right)^{3k+1}
\sigma_{-2k}=P(v,\alpha;\sigma_{2},\dots,\sigma_{-2k+2}),
\end{equation}
where $P$ is a polynomial in all its variables. Thus all coefficients
$\sigma_{-2k}$ can be recursively constructed as analytical functions
of $\alpha\in\mathbb{C}$ defined in the disc centered
at $\alpha=1$ and of radius $1-\beta$.

\begin{remark}
Similarly, one can construct expansions of the form \eqref{sigmas}
which correspond to remaining functions in \eqref{gensol} and \eqref{partsol},
namely, to $(\varphi_2)_\mathrm{I}$,
$(\varphi_2)_\mathrm{II}$, and $(\varphi_2)_{+}$.
According to the Wazow theorem they correspond to
different (more general) solutions of the $\sigma$-form of P6.
\end{remark}

\begin{remark}\label{expsmall}
The general solution
of the $\sigma$-form of P6 \eqref{sigmaform} with $q=0$
is determined by two parameters
in the expansion of the form \eqref{logFrs0}, namely,
by the coefficients of the
terms $\log(1-\alpha)$ and $(1-\alpha)^2$.
Our construction shows that if the coefficient of $\log(1-\alpha)$
is fixed as $s^2$, then all solutions whose
coefficients of $(1-\alpha)^2$, as $s\to \infty$ and
$s/r=\text{const.}$, are of the form
$(r-s)(7s-r)/32+o(s^2)$ have the same asymptotic expansion
\eqref{sigmas} with the leading term given by $(\sigma_2)_{-}$.
Therefore their asymptotics as $s\to\infty$ differ by exponentially
small corrections which are defined by the function of $o(s^2)$.
To find these corrections seems
an interesting problem.
\end{remark}

From \eqref{qzero} and \eqref{sigmas} it follows that
\begin{equation}
\log F_{r,s,0}=-s^2 \varphi -\frac{1}{12}\log s
+\sum_{k=0}^{n} \frac{a_{2k}}{s^{2k}}+O(s^{-2n-2}),
\end{equation}
where the coefficients $a_{2k}=a_{2k}(\alpha,v)$
can be obtained by integration of the corresponding
coefficients $\sigma_{-2k}$ and the integration
constants are fixed by expansion \eqref{logC}.
For $a_0$ we find
\begin{equation}
a_0=\frac{1}{8}\log\left(\frac{4v^2\sqrt\alpha}{2(1+v^2)\sqrt\alpha
-(1-v^2)(1+\alpha)}\right)-\frac{1}{12}\log\frac{1-v^2}{2}+\zeta'(-1).
\end{equation}
Similarly, one can obtain further terms; for example, we
find
\begin{multline}
a_2=\frac{(1-v^2)(1-\sqrt\alpha)^2
\left[2(1+\sqrt\alpha)^4v^4+5(1-\alpha)^2v^2-(1-\sqrt\alpha)^4\right]}
{64\left[2(1+v^2)\sqrt\alpha-(1-v^2)(1+\alpha)\right]^3}
\\
-\frac{1}{8}\left(\frac{1}{8}-\frac{1+v^2}{15}+
\frac{v^2(1+v^2)}{15(1-v^2)^2}\right)
\end{multline}
and
\begin{multline}
a_4=-\frac{(1-v^2)(1-\sqrt\alpha)^2}{256\left[2(1+v^2)\sqrt\alpha
-(1-v^2)(1+\alpha)\right]^6}
\Big[8(1+\sqrt{\alpha})^8(1-3\sqrt{\alpha}+\alpha)v^{10}
\\
+10\left(1-\sqrt{\alpha}\right)^2 \left(1+\sqrt{\alpha}\right)^6
\left(5-46\sqrt{\alpha}+5\alpha\right)v^8
\\
-5(1-\alpha)^4\left(11+106\sqrt{\alpha}+11\alpha\right)v^6
\\
+\left(1-\sqrt{\alpha}\right)^6\left(1+\sqrt{\alpha}\right)^2
\left(1+18\sqrt{\alpha}+\alpha\right)v^4
\\
-\left(1-\sqrt{\alpha}\right)^8
\left(5+14\sqrt{\alpha}+5\alpha\right)v^2
+(1-\sqrt{\alpha})^{10}
\Big]
\\
-\frac{v^6(v^6-4v^4+5v^2-10)}{504(1-v^2)^4}+\frac{31}{16128}.
\end{multline}
Calculations show that the functions $a_{2k}$ become
more and more cumbersome as $k$ increases.
Finally, rewriting the above formulas in terms of
the parameter $\vc$ instead of $\alpha$, see \eqref{vc},
we arrive at Theorem~\ref{TDL}.

\begin{remark}\label{RemMethod}
Our justification of ansatz \eqref{sigmas} in fact includes the following steps:
\begin{enumerate}
\item
Appearance of the formal expansion \eqref{sigmas}
as the simplest possible analytic (in $\alpha$) generalization of
the behavior of $F_{r,s,0}$ for large $r,s$ ($s/r=\text{const.}\equiv v$) at  $\alpha=1$.
\item
Substitution of the ansatz into the $\sigma$-form of P6 equation to prove 
that this ansatz solves the equation.
\item
Application of the Wasow theorem (\cite{W-87}, Th.~36.1) to prove that 
there exists at least one genuine solution of P6
that has asymptotics \eqref{sigmas}.
\item
Proof that the solution obtained via the Wasow theorem 
coincides with the function $\sigma$ related to $F_{r,s,0}$, see \eqref{qzero}, 
based on matching of asymptotics expansions as $\alpha\to1$ 
using the double character of asymptotics \eqref{sigmas} 
with respect to $s\to\infty$ and $\alpha\to1$.
\end{enumerate}
\end{remark}

Some of the items of Remark~\ref{RemMethod} require further comments.

Concerning item (2), using mathematical induction and formulas \eqref{partsol} and 
\eqref{recurr} one can prove that the coefficients $\sigma_{-2k}$, $k=0,1,\dots$, 
are rational functions
of $\sqrt\alpha$ with the poles only at the point $\alpha=\beta$.   
Obviously, these functions are holomorphic with respect to $\alpha$ in any 
simply-connected domain in 
$\mathbb{C}\setminus \{0,\beta\}$.

Concerning item (3), the Wasow theorem does not apply to the 
$\sigma$-form of P6 directly, because it deals with the first-order vector 
ODEs resolved with respect to the derivatives, with holomorphic right-hand sides, and 
a small parameter $\epsilon$. 
This theorem applies instead to the corresponding representation 
of P6 as a Hamiltonian system. The $\sigma$-function 
is intimately related to the Hamiltonian, and is given as a
quintic polynomial in terms of the canonical variables
\cite{O-87}. Conversely, the canonical variables can be written as rational functions 
of $\sigma$, $\partial_\alpha\sigma$, $\partial_\alpha^2\sigma$ and $\alpha$
(see, e.g., Sect.~2.1 in \cite{O-87}). Therefore,  
for the canonical variables of the Hamiltonian system there exist
formal power-like expansions in $\epsilon$, where $\epsilon=1/s$. 
The Wasow theorem also deals with two assumptions, (A) and (B). The assumption 
(A) concerns basically the definition of the domain 
(in the variable $\alpha$) of validity of
the asymptotics; in our case it is evident that it is nonempty, and contains
any segment in the interval $(\beta,1)$. The assumption 
(B) concerns nongeneracy of the matrix of the vector ODE, after, possibly, a proper 
transformation of the canonical variables. In our case it follows from the fact 
that our $\epsilon$-expansions of the canonical variables are defined
uniquely as long as the leading term is given. An example of such transformation for 
the first Painlev\'e equation is given by Wasow 
(see\footnote{There is a misprint after Eq.~(36.40) in \cite{W-87}: 
the assumption (B) is meant rather than (A).}
\cite{W-87}, p. 225). Our case is consired in Appendix~\ref{Appendix-A}.

Concerning item (4), the 
matching of asymptotics is made by using the known asymptotics as
$\alpha\to1$ up to the terms of $O\left((1-\alpha)^3\right)$. 
It is known \cite{J-82} that this asymptotics 
defines uniquely the solution of the $\sigma$-form of P6. Therefore, 
by comparing the expansions as $\alpha\to 1$, 
one can verify that the solutions, a priori defined in a different way, 
actually coincide. 
Strictly speaking, the Wazow theorem guarantees existence of at least one solution with 
the asymptotics \eqref{sigmas}. The other possible solutions may differ by terms decaying 
faster that any integer power of $\epsilon$ (e.g., exponentially small terms).
In our case, because of exact coincidence of the coefficients (no free parameters) 
of the terms $(1-\alpha)^k$, $k=0,1,2$, the exponential small corrections (if any) 
does not influence on the matching procedure. In other words, the corresponding monodromy 
data of the solutions coincide.

\subsection{Asymptotic expansion in the ordered regime}
\label{Sec:Order}

To fix integration constants for the terms of the asymptotic
expansion of the EFP in the ordered region, $v\in(0,u)$ (or,
equivalently, $\alpha\in(0,\beta)$), we use the following
result about the EFP in the limit $\alpha\to0$.

\begin{proposition}\label{Propato0}
As $\alpha\to 0$, the EFP behaves as
\begin{equation}\label{frsq-alpha-0}
F_{r,s,q}=1-\binom{r}{s-1}\binom{r+q}{s+q-1}\alpha^{r-s+1}
+O\left(\alpha^{r-s+2}\right),
\end{equation}
where the standard notation for the binomial coefficients have been used,
$\binom{n}{k}\equiv\frac{n!}{k!(n-k)!}$.
\end{proposition}
\begin{proof}
Consider the $n=1$ term in the sum \eqref{detexp}, which is just the
trace of the our operator $\hat K$, see Proposition~\ref{prop1},
\begin{equation}
\Tr(\hat K)=\oint_{C_0}\frac{(\lambda-\alpha)^r}
{(\lambda-1)^{r+q}\lambda^s}\frac{d\lambda}{2\pi\rmi}
\oint_{C_\infty}\frac{(w-1)^{r+q}w^s}{(w-\alpha)^r(w-\lambda)^2}\frac{dw}{2\pi\rmi}.
\end{equation}
Change variables:
\begin{equation}
\lambda=\alpha\nu,\qquad w=1/\mu.
\end{equation}
Note that since $0<\alpha<1$ contour $C_0\to C_0$, while $C_\infty\to -C_0$, thus
\begin{equation}
\Tr(\hat K)=\alpha^{r-s+1}(-1)^q\oint_{C_0}\frac{(1-\nu)^r}{(1-\alpha\nu)^{r+q}\nu^s}
\frac{d\nu}{2\pi\rmi}
\oint_{C_0}\frac{(1-\mu)^{r+q}}{(1-\alpha\mu)^r\mu^{s+q}(1-\alpha\nu\mu)^2}
\frac{d\mu}{2\pi\rmi}.
\end{equation}
Therefore as $\alpha\to0$ we find
\begin{equation}
\Tr(\hat K)=\alpha^{r-s+1}(-1)^q\oint_{C_0}\frac{(1-\nu)^r}{\nu^s}
\frac{d\nu}{2\pi\rmi}
\oint_{C_0}\frac{(1-\mu)^{r+q}}{\mu^{s+q}}\frac{d\mu}{2\pi\rmi}
\left(1+O(\alpha)\right).
\end{equation}
Now we calculate the above integrals via residues, that gives
$(-1)^{s-1}\binom{r}{s-1}$ and $(-1)^{s+q-1}\binom{r+q}{s+q-1}$,
respectively,
and so we get the second term in the right-hand side of
\eqref{frsq-alpha-0}.

Analogous consideration shows that
$\Tr(\hat K^k)=O(\alpha^{k(r-s+1)})$ as $\alpha\to0$
and thus the further terms in \eqref{detexp} do not contribute
to the leading order of the Taylor series of $F_{r,s,q}$ at $\alpha=0$.
\end{proof}

Setting $q=0$, we have the following result.
\begin{corollary} \label{Cor:ato0}
As $\alpha\to 0$,
\begin{equation}\label{logato0}
\log(1-F_{r,s,0})= (r-s+1)\log \alpha + 2\log\binom{r}{s-1}
+ O(\alpha)
\end{equation}
and
\begin{equation}\label{logFa0}
\log F_{r,s,0}=-\exp\left\{(r-s+1)\log \alpha + 2\log\binom{r}{s-1} + O(\alpha)\right\}.
\end{equation}
\end{corollary}

Our next task is to get asymptotics of the logarithm of
the binomial coefficient in \eqref{logato0} as $s\to\infty$, $r=s/v$.
We recall the known asymptotics  of
logarithm of the Gamma-function as $z\to\infty$
(see, e.g., \cite{B-53}),
\begin{multline}
\log\Gamma(z+a)=
\left(z+a-\frac{1}{2}\right)\log z -z
+\frac{1}{2}\log(2\pi)
\\
+\sum_{n=1}^m
\frac{(-1)^{n+1}B_{n+1}(a)}{n(n+1)z^n}+O(z^{-m-1}),
\qquad |\arg z|<\pi,
\end{multline}
where $B_n(a)$ are the Bernoulli polynomials,
\begin{equation}
B_n(a)=\sum_{k=0}^n \binom{n}{k}B_k a^{n-k}.
\end{equation}
Here $B_k$ are the Bernoulli numbers,
$B_1=-1/2$, $B_{2k+1}=0$, $k=1,2,\dots$, and
$B_{2k}$ are given by \eqref{Bernoulli}.

Using the above formula for
asymptotics of logarithm of the Gamma-function we
arrive at the following result.
\begin{proposition}\label{Cor:logrs}
As $s\to\infty$ and $r=s/v$,
\begin{multline}\label{logB}
2\log \binom{r}{s-1} =- 2\left(\log v +\frac{1-v}{v}\log(1-v)\right) s
-\log s + \log\frac{v^2}{2\pi (1-v)^3}
\\
+\sum_{n=1}^m\frac{B_{2n}v^{2n-1}-(B_{2n}+2n)(\frac{v}{1-v})^{2n-1}
-B_{2n}}
{n (2n-1)s^{2n-1}}
\\
+\sum_{k=1}^m\frac{v^{2k}}{k (1-v)^{2k} s^{2k}}+O(s^{-2m-1}).
\end{multline}
\end{proposition}

Corollary \ref{Cor:ato0} and Proposition \ref{Cor:logrs}
imply that the simplest possible asymptotic ansatz
for the logarithm of $F_{r,s,0}$ as $s\to\infty$, $r=s/v$,
reads
\begin{equation}\label{anlogF}
\log F_{r,s,0}=-\exp\left\{-\chi s-\log s
+\sum_{k=0}^{m}\frac{b_k}{s^k}+O(s^{-m-1})\right\}.
\end{equation}
In view of \eqref{qzero}, the asymptotic ansatz
for the $\sigma$-function, which corresponds to \eqref{anlogF},
is
\begin{equation}\label{sigmaomega}
\sigma= \left(-\frac{(1+v)^2}{4v^2}\alpha+\frac{1}{2v}\right)s^2 +
\rme^{-s\omega}+O(\rme^{-2s\omega_0}),
\end{equation}
where the function $\omega$ admits an asymptotic
expansion in the inverse powers of $s$,
\begin{equation}\label{omega}
\omega=\omega_0 + \omega_{1} s^{-1}+ \omega_{2}s^{-2}+\cdots.
\end{equation}

At first glance, the asymptotic ansatz \eqref{sigmaomega}
is considerably different in comparison with the ansatz for
the $\sigma$-function in
the disordered regime, see \eqref{sigmas}. However, it can be viewed
as the same ansatz but with the function $\varphi\equiv 0$,
$\alpha\in(0,\beta)$. Indeed, as it is mentioned in
Remark~\ref{expsmall}, an asymptotic expansion
for the $\sigma$-function
may contain exponentially small corrections, which are not
indicated in \eqref{sigmas}. At the same time, in case
$\varphi\equiv 0$, the
corresponding function $\sigma_2$
is just given by the coefficient
of the first term in \eqref{qzero}. From the point of view
of the solutions of \eqref{reducedODE}, which govern $\sigma_2$,
this expansion corresponds to the unique
case where both general solutions \eqref{gensol} coincide,
$C_\mathrm{I}=C_\mathrm{II}=-(1+v^2)/4v^2$.
It happens that in this very case
$\sigma_{-2k}\equiv 0$, $k=0,1,\dots$.
Comparing \eqref{qzero} and \eqref{sigmaomega} with \eqref{s1decay}
we see that the function $\chi$ is equal to $\omega_0$,
moreover it should be positive for $\alpha\in(0,\beta)$.

Justification of asymptotic
ansatz \eqref{sigmaomega}, \eqref{omega}
can be obtained analogously as it was
done for ansatz \eqref{sigmas}.
However, here we have an independent proof for
ansatz \eqref{anlogF}
based on the saddle-point method applied to the Fredholm
determinant, presented in Appendix~\ref{Appendix-B}. Therefore in what follows
we consider an explicit construction of the coefficients in the
series \eqref{omega}.

From the $\sigma$-form of P6 it follows that the $\omega$-function
in all orders in $1/s$ obeys the following ODE:
\begin{multline}\label{omegaeq}
\alpha^2(\alpha-1)^2\left(s(\omega')^2-\omega''\right)^2
\\
=\left[(\alpha-1)s\omega'+1\right]
\left[\frac{(1+v)^2\alpha-(1-v)^2}{v^2}s\omega'+\frac{(1+v)^2}{v^2}
\right].
\end{multline}
For example, keeping the leading order in this equation, we find that
the function $\chi\equiv\omega_0$ satisfies equation
\begin{equation}\label{chieq2}
\alpha^2(\alpha-1)(\chi')^2=
\frac{(1+v)^2\alpha-(1-v)^2}{v^2}.
\end{equation}
Since $\alpha\in (0,1)$ and $\chi'$ is real, the
right-hand side of \eqref{chieq2} is negative, therefore
$\alpha\in(0,\beta)$, where $\beta<1$ is given in \eqref{beta}.
Using \eqref{logato0} and \eqref{anlogF}, we note that
$\chi$ has the following asymptotic behavior as $\alpha\to 0$,
\begin{equation}\label{chi-init}
\chi=-\frac{1-v}{v}\log\alpha
+2\left(\log v+\frac{1-v}{v}\log(1-v)\right)+O(\alpha).
\end{equation}
Solving \eqref{chieq2} for $\chi'$, we obtain
\begin{equation}\label{chieq}
\chi'=
-\frac{1}{v\alpha}
\sqrt{\frac{(1-v)^2-(1+v)^2\alpha}{1-\alpha}}
=
-\frac{2}{(1-\sqrt{\beta})\alpha}
\sqrt{\frac{\beta-\alpha}{1-\alpha}},
\end{equation}
where sign minus is chosen to ensure
the correct coefficient of the $\log\alpha$ term in \eqref{chi-init}.
Integrating \eqref{chieq} and fixing the integration constant
according to the second term in \eqref{chi-init}, we obtain
\begin{equation}
\chi=\frac{4}{1-\sqrt{\beta}}
\left\{\sqrt{\beta}
\log\left(
\frac{\sqrt{\beta(1-\alpha)}+\sqrt{\beta-\alpha}}{\sqrt{(1-\beta)\alpha}}
\right)
-\log\left(\frac{\sqrt{1-\alpha}+\sqrt{\beta-\alpha}}{\sqrt{1-\beta}}\right)
\right\}.
\end{equation}
Note that this expression can also be rewritten as
\begin{multline}\label{longchi}
\chi
=2\frac{1+\sqrt{\beta}}{1-\sqrt{\beta}}
\log\left(
\frac{\sqrt{\beta}+\alpha+\sqrt{(1-\alpha)(\beta-\alpha)}}
{(1+\sqrt{\beta})\sqrt{\alpha}}\right)
\\
-2\log\left(
\frac{\sqrt{\beta}-\alpha+\sqrt{(1-\alpha)(\beta-\alpha)}}
{(1-\sqrt{\beta})\sqrt{\alpha}}
\right),
\end{multline}
or, noticing that the numerators of arguments of the logarithms are perfect squares, as
\begin{multline}\label{longchi2}
\chi
=4\frac{1+\sqrt{\beta}}{1-\sqrt{\beta}}
\log\left(
\frac{\sqrt{(1+\sqrt{\alpha})(\sqrt{\beta}+\sqrt{\alpha})}
+\sqrt{(1-\sqrt{\alpha})(\sqrt{\beta}-\sqrt{\alpha})}}
{\sqrt{2(1+\sqrt{\beta})\sqrt{\alpha}}}\right)
\\
-4
\log\left(
\frac{\sqrt{(1-\sqrt{\alpha})(\sqrt{\beta}+\sqrt{\alpha})}
+\sqrt{(1+\sqrt{\alpha})(\sqrt{\beta}-\sqrt{\alpha})}}
{\sqrt{2(1-\sqrt{\beta})\sqrt{\alpha}}}\right).
\end{multline}
These expressions imply that $\chi|_{\alpha=\beta}=0$, and furthermore
$\chi>0$ for $\alpha\in(0,\beta)$, since it can be written, similarly to \eqref{phi-int},
as the integral
\begin{equation}
\chi=\frac{2}{1-\sqrt{\beta}}\int_{\alpha}^{\beta}
\frac{1}{\tilde\alpha}\sqrt{\frac{\beta-\tilde\alpha}{1-\tilde\alpha}}\, \rmd \tilde \alpha.
\end{equation}
Using \eqref{longchi} it can be shown that $\chi=2J(t+1)$, where $J(t)$ is 
the upper tail rate function (for $\gamma=1$)
obtained by Johansson, see Eqs. (2.21) and (2.25) in~\cite{J-00}, and
$t=v^{-1}-1$. Rewriting \eqref{longchi2} in terms of $v$ and $\vc$, see \eqref{vc},
we arrive at \eqref{chiuv}.

After we have determined $\omega_0\equiv \chi$, the
first derivatives of the further coefficients
in the expansion \eqref{omega}, namely $\omega_k'$'s,
can be derived uniquely from \eqref{omegaeq}.
Equations \eqref{qzero}, \eqref{anlogF}, and \eqref{sigmaomega}
implies the following equation for the coefficients $b_k$:
\begin{equation}
\log \alpha+\log (1-\alpha)
+\log (-\chi')+\log\left(1+\sum_{k=0}^\infty\frac{b_k'}{(-\chi') s^{k+1}}\right)
+\sum_{k=0}^\infty \frac{b_k}{s^k}
=-\sum_{k=0}^\infty \frac{\omega_{k+1}}{s^{k}}.
\end{equation}
For example, the first three coefficients $b_k$ are
\begin{equation}
\begin{split}
b_0
&=-\omega_{1}-\log (-\chi')-\log \alpha-\log (1-\alpha),
\\
b_1
&=-\omega_2+\frac{b_0'}{\chi'},
\\
b_2
&=-\omega_3+\frac{b_1'}{\chi'}+\frac{(b_0')^2}{2(\chi')^2}.
\end{split}
\end{equation}
The undermined constants in $\omega_k$ are fixed with the help
of the expansions \eqref{logFa0} and \eqref{logB}.
Using this algorithm, we obtain the following expressions:
\begin{equation}\label{coeffb2}
\begin{split}
b_0
&=\log\left(\frac{\alpha v^2}{2\pi\sqrt{1-\alpha}
\left[(1-v)^2-\alpha(1+v)^2\right]^{3/2}}\right).
\\
b_1
&=-\frac{(1+v^4)(1-\alpha)^2+9(v+v^3)(1-\alpha^2)-2v^2(10\alpha^2-11\alpha+10)}
{6\sqrt{1-\alpha}\left[(1-v)^2-\alpha(1+v)^2\right]^{3/2}},
\\
b_2
&=\frac{v^2}
{(1-\alpha)((1-v)^2-(1+v)^2\alpha)^3}\big[
(\alpha^2+8\alpha+1)(1-\alpha)^2(v^4+1)
\\ &\qquad
-4(1-\alpha)(\alpha^3+1)(v^3+v)+6(1-\alpha)^4v^2+4\alpha(\alpha^2+\alpha+1)v^2\big].
\end{split}
\end{equation}

Clearly,
\begin{equation}
\log(1- F_{r,s,0})=-\chi s-\log s
+\sum_{k=0}^{m}\frac{b_k}{s^k}+O(s^{-m-1}),
\end{equation}
that finishes the proof of Theorem~\ref{TDL2}.

It is interesting to note that \eqref{omegaeq}
can be solved explicitly in terms of the hypergeometric function,
see Appendix~\ref{Appendix-C}.

\section{Conclusion}

The goal of this paper is twofold: to elaborate an algorithm for
construction of full asymptotic expansions for a physically interesting quantity
for an integrable model, and to show that the
theory of P6 can serve as an effective tool for this purpose.

The situation is not absolutely standard since we used an ODE with respect to a
continuous variable to find asymptotic expansions with respect to discrete
parameters. At first glance, Toda or discrete Painlev\'e equations might be more
suitable for this problem. Indeed these equations can be used in
obtaining the leading terms of our asymptotic expansions \cite{CP-13},
however, a perspective of obtaining higher order terms (HOTs) of the expansions 
with the help of the difference or differential-difference equations seems unclear.

Our problem looks like a standard object for asymptotic methods based on the
Riemann-Hilbert conjugation problem or isomonodromy deformations (see
Section~\ref{sec:RH}). These methods however deliver the leading terms of the
corresponding asymptotic expansions, while the explicit formulas for HOTs
require much more deliberate and technically complicated
efforts, so that anyway HOTs are easier to derive via substitution of the
asymptotic expansions into the corresponding integrable (differential,
difference, etc) equation.

In this paper we show that in our case not only HOTs but also the leading terms
can be obtained via the substitution of asymptotic expansions into P6. Since we
use the ODE that plays an auxiliary role with respect to discrete
parameters we need initial conditions for determination of all terms of the
asymptotic expansion rather than just for the leading terms, as it usually happens
in the standard approach. To get these initial conditions one needs to know asymptotic
behavior of the $\tau$-function at critical points of P6. In our case it is asymptotic
behavior of $F_{r,s,0}$ as $\alpha\to1$ and $\alpha\to0$,
Propositions~\ref{prop:al-to-1} and
\ref{Propato0}, respectively. The results stated in these propositions are,
in fact, equivalent to the connection formula for the P6 $\tau$-func\-tion. We 
have found no corresponding result in the literature, probably because our
solution is too special from the point of view of the general theory of P6.

Asymptotics as $\alpha\to0$ is obtained in Proposition~\ref{Propato0} with the
the help of the Fredholm kernel $\hat{K}$, which is the principle manifest of
integrability. At the same time the proof of asymptotic expansion as
$\alpha\to1$ in Proposition~\ref{prop:al-to-1} is obtained by an \emph{ad hoc}
observation, which may not occur for some other analogous problem. We used this
observation because it simplifies our derivation, however we could employ a
more deliberate scheme, we are going to outline below, where $\alpha\to1$ asymptotics
is not required. In principle, the connection problem for
$\tau$-function of P6 is a different and external problem  in comparison to the one
we consider here, so that we can just refer to it. However, there is a way to avoid this
problem, or more precisely, to change it to another one: the connection problem
for the ``corresponding'' $\tau$-function of the second Painlev\'e equation.

The key point of this scheme is the representation of $F_{r,s,q}$ as a
determinant of the integrable Fredholm operator $\hat{K}$. The next step is the
derivation of $\sigma$-form of P6. Our problem is to find asymptotics of
$F_{r,s,q}$ for large variables $r$, $s$, and $q$ in the thermodynamic limit
($s/r=v=\text{const}$) for all values of the parameter $\alpha$, $0<\alpha<1$.
As shown in the paper, it is easy to find asymptotics in the ordered regime,
i.e., in the interval $0<v<u=(1-\sqrt{\alpha})/(1+\sqrt{\alpha})$. This
asymptotics requires $\alpha\to0$ asymptotics and the $\sigma$-form of P6 for
$\Det\hat K$. In Appendix~\ref{Appendix-B} we show that in this case we actually can find
asymptotics directly, without any reference to P6, by the help of the
saddle-point method for the $\Tr\hat{K}$.

Then we observe that in the neighborhood $u\sim v$ our asymptotic expansion is
destroyed (see Theorem~\ref{TDL2}) and the so-called ``boundary'' 
variable $z$ can be introduced,
which, as is easy to guess from the structure of HOTs, reads $z=(u-v)s^{2/3}$.
Substituting it into the $\sigma$-form of P6 instead of $\alpha$ one {\it
obviously} finds that $\sigma(\alpha)\to\sigma_2(z)$, where $\sigma_2(z)=\int
y^2(\tilde z)\,d\tilde z$ is a solution of the $\sigma$-form of the second
Painlev\'e equation (P2). It is well-known that
$\sigma_2(z)=-\int y^2(\tilde z)\,d\tilde z$, where $y(\tilde z)$ is the
corresponding solution of the canonical form of P2,
\begin{equation}\label{eq:P2}
y''=2y^3+\tilde zy+C.
\end{equation}

Let us explain why the appearance of P2 is an obvious thing here. As follows from
the calculations presented in Appendix~\ref{Appendix-B} there is a coalescence of
two simple saddle-points in $\Tr\hat K$ at $u=v$. This is equivalent to
coalescence of two simple saddle-points in the steepest-descent method for the
corresponding RH problem or coalescence of four turning points in the isomonodromy
deformation method. This process is described by the solution of the model
RH-problem corresponding to P2 (see details, say, in \cite{K-94}). Surely, one
can deduce it directly from the corresponding differential equations with the
help of the above-mentioned substitution and obtain one of the limits $P6\to
P2$ (see, say, discussion in introduction to \cite{K-06}). How to choose the
proper solution of equation~\eqref{eq:P2} in the framework of this approach?
The solution of P2 should match with the asymptotics in the ordered case: this
means that in terms of variable $z$ it should be exponentially decaying: thus
the constant $C$ in \eqref{eq:P2} vanishes. It is well-known that there exists
only one-parameter family of solutions of P2 regular on the positive real
semi-axis and exponentially decaying there. The value of this parameter should
be determined via the matching procedure in the domain $z\sim +s^\epsilon$,
$\epsilon>0$, with the leading term of asymptotics in the ordered regime, $v<u$,

On the negative semi-axis solutions of this family can be oscillating or possess
infinite number of poles, depending on the value of the parameter of the family.
There exists also the unique solution that separate these two types of asymptotic
behavior, it can be treated as the solution with the only pole at the point
at infinity, it is known as the Hastings-McLeod solution \cite{HMcL-80}. This
solution grows as $\sqrt{-z}/2$ for $z\to-\infty$.

Since we already know asymptotics of $F_{r,s,q}$ in the disordered regime,
which is regular and non-oscillating, it is obvious that the only suitable solution
of P2 is the Hastings-McLeod solution. However from the formal point of view this fact
should be established by applying the matching procedure described above.

Asymptotics of the Hastings-McLeod solution for $z\sim -s^\epsilon$,
$\epsilon>0$, serves as the ``initial data'' for asymptotics of $F_{r,s,q}$ in
the region $v>u$. Obviously, to use the ``P2-asymptotics'' as the
``initial data'' we have to develop this asymptotic expansion up to all
decaying orders of $z$. From our point of view, this expansion is an important
ingredient of the asymptotic scheme based on matching of asymptotic expansions.
In certain cases, the scheme which can serve as an effective alternative to
other asymptotic methods.

At this stage we just explain that the leading terms of the ``P2-asymptotics'' 
and the expansion constructed
in Theorem~\ref{TDL} do match.
For the leading term we have
\begin{equation}\label{eq:P2-matching-left}
\log F_{r,s,q}\sim-\int_{-\infty}^z(z-\tilde z)y^2(\tilde z)\,
\rmd\tilde z\sim-(v-u)^3s^2/12,
\end{equation}
since $y\sim\sqrt{-z/2}$. This shows that the ``Painlev\'e asymptotics''
in the transition domain matches at least with the leading order
of asymptotics obtained in Theorem~\ref{TDL}. Above we have given an explanation
of appearance of P2 purely in terms of the theory of differential equations and
nonlinear special functions.
For the corresponding matrix model the result is well-known via the appearance of
the ``Airy-kernel'' of the associated integrable integral Fredholm operator
$\hat{K}$ \cite{TW-94}.

Simultaneously, the last asymptotic relation in \eqref{eq:P2-matching-left}
explains the physically interesting phenomenon of the third-order phase transition
mentioned in Introduction: the leading term of the asymptotic
expansion of $\log F_{r,s,q}$ (denoted as $\varphi$ in Theorem~\ref{TDL}) vanishes
at $v=u$, together with the first two derivatives (with respect to $v$ or $u$),
whilst the third derivative does not. In other words, interpreting
$\varphi$ as the change of the free energy per volume in the thermodynamic
limit, $-v^2\varphi/(1+v)^2\equiv\lim_{s,r\to\infty}\log F_{r,s,q}/(r+s)^2$,
one concludes that \eqref{eq:P2-matching-left} implies $\varphi\equiv0$ for $v<u$,
thus reproducing the result of \cite{CP-13} about the
leading term of asymptotics described by the function $\varphi$.

\appendix
\section{The Wasow theorem and the $\sigma$-form of P6}
\label{Appendix-A}

The $\sigma$-forms of Painlev\'e equations does not fit the setup 
of the Wasow theorem [Th.~36.1 in Chap.~IX of  \cite{W-87}], because they are quadratic 
with respect to the second derivative of the $\sigma$-function. Nevertheless, in many 
cases this theorem can be applied to these equations due to the presence of the 
Hamiltonian structure.   

Here we give some more details to the general comments 
given at the end of Sect.~4.2 that support item (3) of Remark 4.3. 
Namely, we show here that the Wasow theorem applies to the  representation 
of P6 as a Hamiltonian system, and due to the one-to-one 
correspondence between the canonical variables and $\sigma$-function  
(as proven by Okamoto \cite{O-87}), it implies the validity of the similar result 
for the $\sigma$-form of P6. We first consider the P6 as a Hamiltonian system in 
a general setup, and next proceed to our case.  

Since here we deal with P6 as a Hamiltonian system, we will rely upon 
results obtained by Okamoto \cite{O-87}. To keep contact with our discussion in the main 
text, we make the following identification 
between parameters $\nu_1,\dots,\nu_4$ of Jimbo-Miwa 
and those $b_1,\dots,b_4$ of Okamoto:
\begin{equation}\label{nus-and-bs}
\begin{split}
\nu_1&=\frac{1}{2}(\theta_\alpha+\theta_\infty)=b_3,
\\
\nu_2&=\frac{1}{2}(\theta_\alpha-\theta_\infty)=b_4,
\\
\nu_3&=-\frac{1}{2}(\theta_0+\theta_1)=-b_1,
\\
\nu_4&=-\frac{1}{2}(\theta_0-\theta_1)=-b_2.
\end{split}
\end{equation}
Note that the parameters $b_1,b_2,b_3,b_4$ in \eqref{nus-and-bs} and the $\sigma$-function 
by Jimbo-Miwa, which we use all over the paper, are in fact 
those denoted in  \cite{O-87} as $b_1^+,b_2^+,b_3^+,b_4^+$ and $h^+$, respectively
(see Eqs. (C.55)--(C.61) of \cite{JM-81}, and Eq. (1.13) and Props. 1.6 and 1.8 
of \cite{O-87}). This means that the canonical variables $q$ and $p$ appearing below 
are related to those of P6 given in Sect.~3.1\footnote{The parameters of P6 
used by Jimbo-Miwa in \cite{JM-81} and by Okamoto in \cite{O-87} are related as 
$\theta_0=\varkappa_0$, $\theta_1=\varkappa_1$, $\theta_\alpha=\theta$, 
and $\theta_\infty=\varkappa_\infty+1$.
} 
by a birational canonical transformation (see Eq. (1.12) of \cite{O-87}).

Given the set of parameters $b_1,\dots, b_4$ and the 
$\sigma$-function, 
the canonical coordinate $q$ and momentum $p$ 
are given by the expressions  
\begin{equation}\label{qandp}
\begin{split}
q&=\frac{1}{2A}[(b_3+b_4)B+(\sigma'-b_3b_4)C],
\\	
p&=\frac{1}{2Aq(q-1)}[-(\sigma'-w_2)B+(w_1\sigma'-w_3)C],
\end{split}
\end{equation}	
where 
\begin{equation}\label{ABC}
A=(\sigma'+b_3^2)(\sigma'+b_4^2),\quad
B=\alpha(\alpha-1)\sigma''+s_1\sigma'-s_3,\quad
C=2(\alpha\sigma'-\sigma)-s_2.
\end{equation}
Here, $s_k$ and $w_k$, $k=1,2,3$, are the $k$-th fundamental symmetric 
polynomials of $\{b_1,b_2,b_3,b_4\}$ and $\{b_1,b_3,b_4\}$, respectively.  
The corresponding system of Hamiltonian equations reads
\begin{equation}\label{Hameqs}
\begin{split}
q'&=\frac{q(q-1)(q-\alpha)}{\alpha(\alpha-1)}
\left(2p-\frac{\theta_0}{q}-\frac{\theta_1}{q-1}
-\frac{\theta_\alpha}{q-\alpha}\right),
\\
p'&=\frac{1}{\alpha(\alpha-1)}\bigg\{
\left[-3q^2+2(\alpha+1)q-\alpha\right]p^2
\\ &\quad
+\left[\theta_0(2q-1-\alpha)+\theta_1(2q-\alpha)+\theta_\alpha(2q-1)\right]p-
\frac{(\theta_0+\theta_1+\theta_\alpha)^2-\theta_\infty^2}{4}
\bigg\}.
\end{split}
\end{equation}
Note, that eliminating $p$, one finds that $q$ solves the canonical form of P6
with $a=\tfrac{1}{2}\theta_\infty^2$, $b=-\tfrac{1}{2}\theta_0^2$,
$c=\tfrac{1}{2}\theta_1^2$, and $d=-\tfrac{1}{2}\theta_\alpha(\theta_\alpha+2)$,
in agreement with the birational canonical transformation 
(which can be identified as the change $\theta_\alpha\mapsto\theta_\alpha+1$, 
$\theta_\infty\mapsto\theta_\infty+1$ in the set of monodromy 
parameters of P6 given in Sect.~3.1).

Now we are ready to consider our case. 
In Sect.~4.2 we impose the following relations on the parameters: 
\begin{equation}
\nu_1=\nu_3=-\frac{r+s}{2},\qquad
\nu_2=-\nu_4=-\frac{r-s}{2}.
\end{equation}
Hence,  
$b_3=-b_1=-\nu_1$ and $b_4=b_2=\nu_2$, and so we have 
\begin{align}
s_1&=2\nu_2,& 
s_2&=-\nu_1^2+\nu_2^2,&
s_3&=-2\nu_1^2\nu_2,
\\
w_1&=\nu_2,& 
w_2&=-\nu_1^2,&
w_3&=-\nu_1^2\nu_2,
\end{align}
We find it convenient to set
\begin{equation}
\nu_1=sv_1,\qquad \nu_2=sv_2, 
\end{equation}
where
\begin{equation}
v_1\equiv-\frac{1+v}{2v},\qquad
v_2\equiv-\frac{1-v}{2v},\qquad
v=\frac{s}{r}.
\end{equation}
We recall that $s$ is a large parameter, while $v_1$ and $v_2$
are finite.

Let us map the system \eqref{Hameqs} into the form considered by Wasow.
Consider the formal asymptotic expansion for the 
$\sigma$-function \eqref{sigmas}, which is uniquely defined by its first term, 
\begin{equation}
\sigma=s^2\left(\frac{v_1^2+v_2^2}{2}-2v_1v_2\sqrt{\alpha}\right)
+O(1),\qquad s\to\infty.
\end{equation}
Using it and equations \eqref{qandp} and \eqref{ABC}, 
one finds the corresponding formal expansions for $q$ and $p$ in 
$1/s$. The first terms reads
\begin{equation}\label{qpexpansion}
\begin{split}
q&=-\sqrt{\alpha}+s^{-1}\frac{(\alpha-1)\sqrt{\alpha}(v_1+v_2)}
{4(v_1\sqrt{\alpha}-v_2)(v_2\sqrt{\alpha}-v_1)}+O(s^{-2}),
\\
p&=s\frac{v_1\sqrt{\alpha}-v_2}{\sqrt{\alpha}(\sqrt{\alpha}+1)}+
\frac{(\sqrt{\alpha}-1)\left[(v_1+2v_2)\sqrt{\alpha}+v_2\right]}
{4\sqrt{\alpha}(\sqrt{\alpha}+1)(v_2\sqrt{\alpha}-v_1)}+O(s^{-1}).
\end{split}
\end{equation}
Introduce a 2-component vector function $\vec y=\vec y(\alpha)$, 
\begin{equation}
\vec y=
\begin{pmatrix}
y_1 \\ y_2
\end{pmatrix},\qquad
q=y_1,\qquad p=s y_2,
\end{equation}
where this transformation is motivated by the formal expansion
\eqref{qpexpansion}. Taking into account that in our case 
$\theta_0=-\theta_\infty=-\nu_1+\nu_2$ and
$\theta_1=-\theta_\alpha=-\nu_1-\nu_2$, 
we can write the Hamiltonian system \eqref{Hameqs} in the following form:
\begin{equation}\label{Wasow-vector}
\epsilon \frac{\rmd}{\rmd \alpha}\vec y = \vec f(\alpha;\vec y),
\qquad
\vec f(\alpha;\vec y)=
\begin{pmatrix}
f_1(\alpha;\vec y) \\ f_2(\alpha;\vec y)
\end{pmatrix},
\end{equation}
where $\epsilon\equiv 1/s$ and 
\begin{equation}
\begin{split}
f_1(\alpha;\vec y)&=\frac{y_1(y_1-1)(y_1-\alpha)}{\alpha(\alpha-1)}
\left(2y_2+\frac{v_1-v_2}{y_1}-\frac{(v_1+v_2)(\alpha-1)}{(y_1-1)(y_1-\alpha)}\right),
\\
f_2(\alpha;\vec y)&=\frac{y_2}{\alpha(\alpha-1)}
\big\{
\big[-3y_1^2+2(\alpha+1)y_1-\alpha\big]y_2
-2(v_1-v_2)y_1+2v_1\alpha-2v_2
\big\}.
\end{split}
\end{equation}
The Wasow theorem deals with the system
of the form \eqref{Wasow-vector}\footnote{Note, that here 
we meet a special case of the Wasow theorem, where 
$\vec f(\alpha,\vec y)$ is independent of $\epsilon$.
}
and contains two additional conditions, called in \cite{W-87}
assumption (A) and assumption (B). We are going to verify
that these conditions are valid in our case.

To verify assumption (A), see p. 218 in \cite{W-87}, we first assume that 
$\alpha\in\overline D$ where $\overline D$ is any closed disk in 
$\mathbb{C}\setminus \{0,\beta,1,\infty\}$. The case of our particular interest
is where $\alpha$ belongs to any closed segment in the interval $(\beta,1)$. 
Assumption (A) requires also existence of a holomorphic vector function $\vec \phi$ on 
$\overline D$ such that 
\begin{equation}\label{phi}
\vec f(\alpha;\vec \phi)=0.
\end{equation}
In our case 
\begin{equation}\label{phifun}
\vec \phi=
\begin{pmatrix}
\phi_1 \\ \phi_2
\end{pmatrix},\qquad
\phi_1=
-\sqrt{\alpha},\qquad
\phi_2= 
\frac{v_1\sqrt{\alpha}-v_2}{\sqrt{\alpha}(\sqrt{\alpha}+1)}.
\end{equation}
Note that $\vec \phi$ is holomorphic for $\alpha\in \overline D$. 
Thus, assumption (A) is verified.

To verify assumption (B), see p. 219 in \cite{W-87},
we define a new vector function 
\begin{equation}
\vec y_*=\vec y -\vec \phi,
\end{equation}
and rewrite \eqref{Wasow-vector} as follows:
\begin{equation}\label{neweq}
\epsilon \frac{\rmd}{\rmd \alpha}\vec y_* = 
-\epsilon \frac{\rmd}{\rmd \alpha}\vec \phi
+\mathcal{A}\vec y_*+\vec g(\alpha;\vec y_*).
\end{equation}
Equation \eqref{neweq}
exhibits the constant (with respect to $\vec y_*$), linear and 
nonlinear terms of the system of equations 
\eqref{Wasow-vector} about 
the zero-order (in $\epsilon$) solution 
$\vec \phi$; the vector function $\vec g(\alpha;\vec y_*)$ satisfies
$\partial_{y_{*,k}}\vec g(\alpha;\vec y_*)|_{\vec y_*=0}=0$, $k=1,2$. 
Equation \eqref{neweq} is a special case of Eqn.~(36.9) in \cite{W-87}.
The matrix $\mathcal{A}$ reads:
\begin{equation}
\mathcal{A}=
\begin{pmatrix}
-2\dfrac{v_1\sqrt{\alpha}-v_2}{\alpha(1-\sqrt{\alpha})}
& 2\dfrac{1+\sqrt{\alpha}}{1-\sqrt{\alpha}}
\\[10pt]
-2\dfrac{\big[v_1\alpha(2+\sqrt{\alpha})-v_2(1+2\sqrt{\alpha})\big]
(v_1\sqrt{\alpha}-v_2)}{\alpha^2(1-\alpha)(1+\sqrt{\alpha})^2}
&
2\dfrac{v_1\sqrt{\alpha}-v_2}{\alpha(1-\sqrt{\alpha})}
\end{pmatrix}.
\end{equation}
The determinant evaluates as 
\begin{equation}
\det \mathcal A = 
4\frac{(v_1\sqrt{\alpha}-v_2)(v_2\sqrt{\alpha}-v_1)}{\alpha^{3/2}(1-\alpha)^2}.
\end{equation}
Recall that the disordered regime which we consider in Sect. 4.2 corresponds
to $\alpha\in(\beta,1]$, where $\beta=(\tfrac{1-v}{1+v})^2=v_2^2/v_1^2$. Thus, 
$\det \mathcal{A}\ne 0$ for all values of $\alpha$ of this regime, and  
assumption (B) of the Wasow theorem is fulfilled.

Hence, the Wasow theorem applies to the 
system \eqref{Wasow-vector}, since both assumptions (A) and (B) 
of the Wasow theorem are satisfied in 
the disordered regime. 
Therefore, from the Wasow theorem it follows that there exists a genuine solution 
$q$, $p$ of the system \eqref{Hameqs} with the asymptotic expansions \eqref{qpexpansion}. 
These expansions are uniform on any closed segment of $(\beta,1)$ and all power-terms 
in $\epsilon$ can be constructed uniquely via the explicitly given zero-order
solution \eqref{phifun}. Due to the one-to-one correspondence
between the canonical variables 
$q$, $p$ and the $\sigma$-function, see Eqn.~(1.6) and Prop.~1.2 in \cite{O-87}, 
there exist genuine solution $\sigma$-function with the formal expansion \eqref{sigmas}.

Note that the switch to the Hamiltonian system $q$ and $p$ has mostly 
theoretical rather than practical sense since from the construction above 
it is not obvious that the resulting expansion for the $\sigma$-function 
contains only even powers of $s$. However, this fact is obvious 
due to the $\sigma$-form of P6 \eqref{sigmaform}.

\section{Saddle-point approach to asymptotic expansion of the EFP in the ordered regime}
\label{Appendix-B}

As in the main text we assume here that  $q=0$, which is not crucial for the
approach discussed below. As mentioned in the main part of the paper the
complete asymptotic expansion of $F_{r,s,q}$ can be presented as transseries
with a finite number (actually, of $s$) series. The transseries
expansion of $\log F_{r,s,q}$ consists of infinite number of series; the $n$-th
series in this expansion is generated by the term $\Tr(\hat K^n)$. Here we
consider only construction of the first series of the complete transseries
expansion, i.e., the series generated by $\Tr(\hat K)$; the other series can be
studied analogously, however the corresponding calculations are more involved as
they are related with the saddle-point method for integrals with $n+1$ variables.

The main object of our studies here is the following integral:
\begin{equation}\label{eq:trace-K-nu-mu-q=0}
\Tr(\hat K)=-\frac{\alpha^{r-s+1}}{(2\pi)^2}
\oint_{C_0}\oint_{C_0}\frac1{(1-\alpha\nu\mu)^2}\frac{(1-\nu)^r}{(1-\alpha\nu)^r}
\frac{(1-\mu)^r}{(1-\alpha\mu)^r}\frac{d\nu\,d\mu}{\nu^s\mu^s}.
\end{equation}
We rewrite it in the following way
\begin{equation}\label{eq:trace-K-nu-mu-S}
\Tr(\hat K)=
-\frac{\alpha^{r-s+1}}{(2\pi)^2}\oint_{C_0}\oint_{C_0}e^{sS(\nu,\mu)}f(\nu,\mu)d\nu\,d\mu,
\end{equation}
where
\begin{equation}\label{eq:S-def}
S(\nu,\mu)\equiv S(\nu)+S(\mu), \qquad
S(\nu)=
\frac{1}{v}\log\frac{1-\nu}{1-\alpha\nu}-\log(-\nu),
\end{equation}
and
\begin{equation}\label{eq:f-def}
f(\nu,\mu)\equiv\frac1{(1-\alpha\nu\mu)^2},
\end{equation}
Note that in the definition of $S(\nu,\mu)$ we used the fact
that $s$ is an integer and we recall that $v=s/r$.

Since the variables $\nu$ and $\mu$ are separated in the function $S(\nu,\mu)$, one can
successively apply the usual saddle-point method
for a single variable with respect to each of these two variables. However,
it is easier to employ the saddle-point method for multiple integrals. We follow below
the version of this method developed by Fedoryuk \cite{F-77}.

Saddle points are defined by the system of equations
$\frac\partial{\partial\nu} S=\frac\partial{\partial\mu} S=0$.
There are four saddle points with the coordinates
$(\nu_-,\mu_-)$, $(\nu_-,\mu_+)$, $(\nu_+,\mu_-)$, and $(\nu_+,\mu_+)$, where
\begin{equation}
\nu_\pm=\mu_\pm=-\frac{1-\alpha}{4\alpha v}
\left(\sqrt{1-v\frac{1+\sqrt{\alpha}}{1-\sqrt{\alpha}}}
\pm\sqrt{1-v\frac{1-\sqrt{\alpha}}{1+\sqrt{\alpha}}}\right)^2.
\end{equation}
In the following we assume that
\begin{equation}\label{eq:u-def}
0<v<u\equiv\frac{1-\sqrt{\alpha}}{1+\sqrt{\alpha}}
\end{equation}
This inequality, together with the condition $0<\alpha<1$, implies that
\begin{equation}
\nu_+<-1/{\alpha}^{1/2}<\nu_-<0.
\end{equation}
To  prove it one has to notice that $\nu_+\nu_-=1/\alpha$ where both factors
in the left-hand side are negative. Using this relation we find that
\begin{equation}\label{eq:S+S}
S(\nu_-)+S(\nu_+)=-\frac{1-v}v\log\alpha>0.
\end{equation}
If we denote $y=-\nu$ and consider $\tfrac{d}{dy}S(-y)$ for $y>0$, then
it is easy to notice that $\tfrac{d}{dy}S(-y)<0$,
for $0<y<y_-$ and $y>y_+$, and  $\tfrac{d}{dy}S(-y)>0$, for $y_-<y<y_+$, where
$y_\pm=-\nu_\pm$. Therefore,
\begin{equation}
S(\nu_+,\mu_+)>S(\nu_+,\mu_-)=S(\nu_-,\mu_+)>S(\nu_-,\mu_-).
\end{equation}
In other words, the saddle point surface should not contain the points $(\nu_+,\mu_+)$,
$(\nu_-,\mu_+)$, and $(\nu_+,\mu_-)$.

Thus, asymptotic expansion consists of the contribution related with the saddle point $(\nu_-,\mu_-)$.
We use the general formula for
this contribution obtained by Fedoryuk \cite{F-77},
which in our case can be written as follows,
\begin{equation}\label{eq:S-gen-asympt}
\Tr(\hat K)=\frac{\alpha^{r-s+1}}{2\pi}\frac{e^{2sS(\nu_-)}}{s\,S_{\nu\nu}''(\nu_-)}
\sum_{k=0}^{\infty}R_k(s;\alpha,v),
\end{equation}
where
\begin{equation}
R_k(s;\alpha,v)=\frac1{k!}\left(\frac{\Delta}
{2s\,S_{\nu\nu}''(\nu_-)}\right)^k\left(f(\nu,\mu)e^{sS_1(\nu,\mu;\nu_-,\mu_-)}\right)
\Big|_{(\nu,\mu)=(\nu_-,\mu_-)},
\end{equation}
and $\Delta$ is the Laplacian,
\begin{equation}
\Delta=\frac{\partial^2}{\partial\nu^2}+\frac{\partial^2}{\partial\mu^2}.
\end{equation}
The function $S_1(\nu,\mu;\nu_-,\mu_-)$ is
\begin{equation}
S_1(\nu,\mu;\nu_-,\mu_-)=S(\nu,\mu)-S(\nu_-,\mu_-)-
\frac{1}{2}S_{\nu\nu}''(\nu_-)\left((\nu-\nu_-)^2+(\mu-\mu_-)^2\right),
\end{equation}
where
\begin{equation}
S_{\nu\nu}''(\nu_-)=\frac{(1-v)\nu_{-}+1+v}{\nu_{-}^2(1-\nu_{-})}.
\end{equation}
It is important to mention that \eqref{eq:S-gen-asympt} produces
exponential decay for large $s$. To prove this,
we rewrite identity \eqref{eq:S+S} as
\begin{equation}
S(\nu_-)+S(\nu_+)+\frac{1-v}{v}\log\alpha=0
\end{equation}
and employ the inequality $S(\nu_-)<S(\nu_+)$, that yields
\begin{equation}
2S(\nu_-)+\frac{1-v}{v}\log\alpha<0.
\end{equation}

It is interesting to note that $S(\nu_-)>0$. To show this let us consider $S(-y)$ for positive $y$.
Since for $0<y<1$ both terms (see the second equation in \eqref{eq:S-def})
are positive, $S(-y)>0$ for these values of $y$. At the same time
\begin{equation}
S(-y)=-\log y-\frac{1}{v}\log\alpha+ O\left(1/y\right)<0,\qquad
y\to\infty,
\end{equation}
so that all positive zeros of $S(-y)$ should be greater than 1.
On the other hand, it is easy to see that $S(1/\sqrt{\alpha})>0$,
therefore one zero is greater than $y_+$ and the others (if any) should be
located in the interval $(1,1/\sqrt{\alpha})$.
To prove that actually $S(-y)>0$ for $y\in(1,1/\sqrt{\alpha})$ consider
the following change of variables, $y=1/\alpha^\epsilon$, where
$0<\epsilon\leq1/2$. In this notation $S(-y)$ reads
\begin{equation}
S(-y)=\frac{1}{v}\log\frac{1+\alpha^\epsilon}{1+\alpha^{1-\epsilon}}
+\frac{1-v}{v}\log\frac1{\alpha^\epsilon}>0.
\end{equation}

Now we study expansion \eqref{eq:S-gen-asympt} in more detail.
Let us note that \eqref{eq:S-gen-asympt}
does not represent an asymptotic series because of differentiation
of the function $\exp(sS_1)$ with respect to the large parameter $s$. To get the
asymptotic series one has to rearrange summation
in a proper way, such that the series would represent an expansion
in the decaying powers of $s$. Let us write,
\begin{equation}
R_k(s;\alpha,v)=\sum_{l=0}^{\left[\frac{2k}3\right]}
\frac{R_{k,l}(\alpha,v)}{s^{k-l}},
\end{equation}
where $\left[\frac{2k}3\right]$ is the entire part of $\frac{2k}3$. Note that
\begin{equation}
R_0(s;\alpha,v)=R_{0, 0}(\alpha,v)=f(\nu_-,\mu_-).
\end{equation}
Thus the asymptotic expansion of $\Tr(\hat K)$ reads
\begin{equation}\label{eq:S-gen-asympt1}
\Tr(\hat K)=\frac{\alpha}{2\pi}\frac{f(\nu_-,\mu_-)}{S^{''}_{\nu\nu}(\nu_-)}\,
\exp\left\{s\left(2S(\nu_-)+\frac{1-v}v\log\alpha\right)-\log s\right\}
\sum_{p=0}^{\infty}\frac{\hat b_p}{s^p},
\end{equation}
where
\begin{equation}\label{eq:ap-def}
\hat b_p=\frac1{f(\nu_-,\mu_-)}\sum_{l=0}^{2p}R_{p+l,l}(\alpha,v).
\end{equation}
In particular,
\begin{align}
\hat b_0&=1,
\\
\hat b_1&=\frac{R_{1,0}(\alpha,v)+R_{2,1}(\alpha,v)+R_{3,2}(\alpha,v)}{f(\nu_-,\mu_-)},
\\
\hat b_2&=\frac{R_{2,0}(\alpha,v)+R_{3,1}(\alpha,v)+R_{4,2}(\alpha,v)
+R_{5,3}(\alpha,v)+R_{6,4}(\alpha,v)}{f(\nu_-,\mu_-)}.
\end{align}

To evaluate $\hat b_1$, we note that
\begin{align}
R_{1,0}(\alpha,v)&=\frac{\Delta f(\nu_-,\mu_-)}{2S''(\nu_-)},
\\
R_{2,1}(\alpha,v)&=
-\frac{1}{4S''(\nu_-)^2}\left(S^{(4)}(\nu_-)+2S^{(3)}(\nu_-)D\right)f(\nu_-,\mu_-),
\\
R_{3,2}(\alpha,v)&=\frac5{12}\frac{S'''(\nu_-)^2f(\nu_-,\mu_-)}{S''(\nu_-)^3},
\end{align}
where
\begin{equation}
D=\frac{\partial}{\partial\nu}+\frac{\partial}{\partial\mu},
\end{equation}
and we use the convention that $Df(\nu_-,\mu_-)\equiv Df(\nu,\mu)
\big|_{\nu=\nu_-,\mu=\mu_-}$. Using the above formulas, we obtain
\begin{equation}
\hat b_1=-\frac{(1-\alpha)^2(v^4+1)+9(1-\alpha^2)(v^3+v)-2(10\alpha^2-11\alpha+10)v^2}
{6\sqrt{1-\alpha}((1-v)^2-(1+v)^2\alpha)^{3/2}}.
\end{equation}
Similarly, to compute $\hat b_2$, we note that
\begin{align*}
R_{2,0}(\alpha,v)&
=\frac{\Delta^2 f(\nu_-,\mu_-)}{8S''(\nu_-)^2},
\\
R_{3,1}(\alpha,v)&
=-\frac{1}{24S''(\nu_-)^3}\Big(S^{(6)}(\nu_-)+3S^{(5)}(\nu_-)D+9S^{(4)}(\nu_-)\Delta
\\  & \qquad
+2S^{(3)}(\nu_-)\big(6\Delta D-D^3\big)\Big)f(\nu_-,\mu_-),
\\
R_{4,2}(\alpha,v)&
=\frac{1}{96S''(\nu_-)^4}\Big(19S^{(4)}(\nu_-)^2
+28S^{(3)}(\nu_-)S^{(5)}(\nu_-)
\\  & \qquad
+76S^{(3)}(\nu_-)S^{(4)}(\nu_-)D
+4S^{(3)}(\nu_-)^2\big(17\Delta+3D^2\big)\Big)f(\nu_-,\mu_-),
\\
R_{5,3}(\alpha,v)&
=-\frac5{48}\frac{S^{(3)}(\nu_-)^2}
{S''(\nu_-)^5}\left(11S^{(4)}(\nu_-)+8S^{(3)}(\nu_-)D\right)f(\nu_-,\mu_-),
\\
R_{6,4}(\alpha,v)&=\frac{205}{288}\frac{S^{(3)}(\nu_-)^4f(\nu_-,\mu_-)}{S''(\nu_-)^6}.
\end{align*}
Using these formulas we obtain
\begin{multline}
\hat b_2
=\frac1{72(1-\alpha)[(1-v)^2-(1+v)^2\alpha]^3}
\big[(1-\alpha)^4(v^8+1)
\\
+18(1+\alpha)(1-\alpha)^3(v^7+v)
+(113\alpha^2+782\alpha+113)(1-\alpha)^2(v^6+v^2)
\\
-18(1-\alpha^2)(35\alpha^2-36\alpha+35)(v^5+v^3)
\\
+12(83(\alpha^4+1)-194(\alpha^3+\alpha)+321\alpha^2)v^4\big].
\end{multline}
This result can be also written as follows:
\begin{equation}
\hat b_1=b_1,\qquad
\hat b_2=\frac{\hat b_1^2}2+b_2.
\end{equation}
Here $b_1$ and $b_2$ are given by \eqref{coeffb2}, where they are rewritten in
terms of $u$ (see Equation~\eqref{eq:u-def}).

\section{Explicit form for the exponentially small correction in the ordered regime}
\label{Appendix-C}

As explained in Appendix~\ref{Appendix-B} the complete asymptotic expansion in
the ordered regime can be written in terms of transseries. The first series is
defined by $\Tr\hat K$ and given by the double
integral~\eqref{eq:trace-K-nu-mu-q=0}. At the same time our construction of
asymptotics  presented in Section~\ref{Sec:Order} makes it possible to obtain another
representation for this series. We recall that in terms of the
$\sigma$-function the corresponding series is written in \eqref{sigmaomega} as
$e^{-s\omega}$, where the function $\omega$ is a solution of
ODE \eqref{omegaeq}.
The general solution of \eqref{omegaeq} can be found
explicitly, namely, it has the following form
\begin{equation}\label{eq:omega-prime}
s\omega'=\frac{z^2-\tilde{a}}{z^2-\tilde{b}-(z^2-\tilde{a})\alpha},\qquad
\tilde{a}=\left(\frac{1+v}v\right)^2,\quad
\tilde{b}=\left(\frac{1-v}v\right)^2,
\end{equation}
where the function $z$ reads
\begin{equation}\label{eq:w-log-psi}
z=\pm\frac{2\alpha}s\frac{\rmd}{\rmd\alpha}\log\psi,
\end{equation}
with
\begin{multline}
\psi=C_1\alpha^{s(1-v)/2v}
\Ftwoone{s/v}{-s}{1+s(1-v)/v}{\alpha}
\\
+C_2
\alpha^{-s(1-v)/2v}
\Ftwoone{s}{-s/v}{1-s(1-v)/v}{\alpha}.
\end{multline}
Here $C_1$ and $C_2$ are the constants of integration, which, in fact, are
functions of $s$ and $v$.

Now we have to choose the solution which corresponds to $F_{r,s,0}$. As follows
from \eqref{qzero} and \eqref{sigmaomega},
\begin{equation}\label{eq:F-rs0-int}
\log F_{r,s,0}=-\int_0^\alpha \rme^{-s\omega}\frac{\rmd\alpha}{\alpha(1-\alpha)}
+O\left(\rme^{-2s\omega}\right).
\end{equation}
Using asymptotics as $\alpha\to0$ of the Gauss hypergeometric function and
comparing it with the expansion~\eqref{logFa0} we find that $C_2=0$ and can fix
the constant of integration while finding $\omega$ from the first
equation in \eqref{eq:omega-prime}. Before presenting our final result we give an
intermediate formula for $s\omega'$ which can be obtained by manipulation with
\eqref{eq:omega-prime} and \eqref{eq:w-log-psi} at $C_2=0$,
\begin{equation}
-s\omega'=\frac{\rmd}{\rmd\alpha}
\log\left[\alpha^{s\frac{1-v}v}
\left(z^2-\tilde{b}-(z^2-\tilde{a})\alpha\right)
\Ftwoone{s/v}{-s}{1+s(1-v)/v}{\alpha}
\right].
\end{equation}
After integration and substitution of $s\omega$ into
\eqref{eq:F-rs0-int} with yet indeterminate constant of integration,
application of the formulas for differentiation of the hypergeometric function
\cite{B-53}, determination of the integration constant as explained above, and
some further simplifications, we arrive at the following representation of
$\log F_{r,s,0}$ in the ordered regime,
\begin{multline}\label{eq:F-rs0-new-int}
\log F_{r,s,0}=-s^2\binom{r}{s}^2
\int_0^\alpha
\bigg\{
\frac{1}{1-\alpha}
\left[\Ftwoone{s/v}{-s}{1+s(1-v)/v}{\alpha}\right]^2
\\
-\frac{s(1-v)}{s(1-v)+v}
\Ftwoone{s/v}{-s}{s(1-v)/v}{\alpha}
\Ftwoone{1+s/v}{1-s}{2+s(1-v)/v}{\alpha}
\bigg\}
\alpha^{s(1-v)/v}\,\rmd\alpha
\\
+O\left(\rme^{-2s\omega}\right).
\end{multline}
To see that this integral, in fact, is exponentially decaying as $s\to\infty$,
one has to apply the saddle-point method to the corresponding integral
representation of the hypergeometric function. Contrary to considerations
of Appendix~\ref{Appendix-B}, the asymptotic analysis of the
representation \eqref{eq:F-rs0-new-int} requires just the
saddle-point method applied to a single integral. Note that the correction term in
\eqref{eq:F-rs0-new-int} has the order of the square of the first term.

\bibliography{pvi_bib}
\end{document}